%% file: sel.tex
\documentclass[11pt]{article}
\usepackage[margin=1in]{geometry} 
\usepackage{graphicx}
\usepackage{xcolor}
\usepackage{float}
\usepackage{hyperref}
\usepackage{natbib}
\usepackage{notation}
\usepackage{setspace}

\usepackage[most]{tcolorbox}
\usepackage{mathabx}
\newtcolorbox[auto counter]{mybox}[2][]{float,title={\textcolor{black}{Box~\thetcbcounter: #2}},#1, 
colframe=white!70!gray}
\usepackage{setspace}

\usepackage{cleveref}
\crefname{figure}{Figure}{Figures}
\crefname{equation}{}{}
\crefname{definition}{Definition}{Definitions}
\crefname{corollary}{Corollary}{Corollaries}
\crefname{proposition}{Proposition}{Propositions}
\crefname{theorem}{Theorem}{Theorems}
\crefname{example}{Example}{Examples}
\crefname{remark}{Remark}{Remarks}
\crefname{principle}{Principle}{Principles}
\crefname{lemma}{Lemma}{Lemmata}
\crefname{claim}{Claim}{Claims}
\crefname{table}{Table}{Tables}
\crefname{section}{Section}{Sections}
\crefname{subsection}{Subsection}{Subsections}
\crefname{subsubsection}{Subsection}{Subsection}
\crefname{assumption}{Assumption}{Assumptions}
\crefname{appendix}{Appendix}{Appendices}

\newcounter{mytcbcounter}
\newtcolorbox[use counter=mytcbcounter]{examplebox}[2][]{
    breakable,
    title={\textbf{\textcolor{black}{Example~\themytcbcounter: #2}}},
    #1, 
    colframe=blue!25!white,
    colback=blue!0!white
    }

\title{Inference conditional on selection: a review}

\author{
Anna Neufeld$^{1}$, Ronan Perry$^{2}$, Daniela Witten$^{2,3}$\\
\\
$^{1}$Department of Statistics, Williams College\\
$^{2}$Department of Statistics, University of Washington\\
$^{3}$Department of Biostatistics, University of Washington\\
}
\date{}

\begin{document}
\setstretch{1.6}

\maketitle

\input{sections/abstract}
\input{sections/sec1_intro}

\input{sections/sec2_new_coverage}

\input{sections/sec3_methods}

\input{sections/sec4_methods_2}

\input{sections/sec5_sim}
\input{sections/sec6_real}
\input{sections/sec7_disc}

\bibliographystyle{plainnat}
\bibliography{sel}

\appendix
\input{sections/zAppendix}

\end{document}

%% file: sections/abstract.tex
\begin{abstract}
In this article, we review \emph{selective inference}, a set of techniques for \emph{inference when the statistical question asked is a function of the data}. This setting often arises in contemporary scientific workflows,  where hypotheses and parameters may be selected from the data, rather than specified in advance. In this setting, classical inferential techniques do not achieve ``classical" guarantees, such as nominal coverage of confidence intervals.  We consider three examples for which selective inference solutions are required: inference on a ``winner", inference on the mean of a region in a regression tree, and inference on the difference in means between a pair of clusters.  We argue that \emph{conditional} guarantees are of scientific interest in such settings. We then review and draw connections between several approaches that provide such guarantees. Finally, we illustrate these approaches in simulation and through an application to single-cell RNA sequencing data. 
\end{abstract}

%% file: sections/sec1_intro.tex
\section{Introduction}
\label{sec1}

In classical statistics, models, hypotheses, and parameters are specified before the data are observed. By contrast, modern scientific practice often involves exploring the data before selecting models, hypotheses, or parameters that appear promising.  When traditional statistical methods such as t-tests or Wald intervals are applied without accounting for this data-driven selection, they fail to provide guarantees such as Type 1 error control or nominal coverage. This issue has been partially blamed for the so-called replication crisis in science 
\citep{benjamini2020selective}.

The practice of naively using the same data both to select a statistical question, and to answer it, is sometimes referred to as \emph{double dipping}. Concerns about double dipping date back at least to the 1940s, when \cite{koopmans1949koopmans} noted, ``we constantly take hints from the data regarding the choice of hypotheses to be tested from the same data--although we know that the degree of confidence to be placed in the test is affected by that practice." Decades later, \cite{breiman1992little} called this practice ``a quiet scandal in the statistical community."  Despite such warnings, double dipping remains common in modern scientific practice, as recently documented in the fields of neuroscience \citep{kriegeskorte2009circular}, genomics \citep{lahnemann2020eleven}, and ecology \citep{campbell2021consequences}. It may be that the statistical community has only itself to blame for the pervasive nature of double dipping: as noted in Chapter 1 of \cite{kuchibhotla2020unified}, introductory statistics textbooks routinely teach classical p-values for linear model coefficients while simultaneously encouraging iterative model selection (i.e., variable selection via stepwise regression), without emphasizing that such p-values are invalid after this iterative process.

Of course, one could avoid double dipping by simply restricting attention to pre-specified hypotheses. But as  \cite{tukey1969analyzing} remarked, ``to concentrate on confirmation, to the exclusion or submergence of exploration, is an obvious mistake. Where does new knowledge come from? How can an undetected criminal be put on trial?" The field of \emph{selective inference} aims to provide tools for valid inference on parameters that were selected based on data exploration. In this paper, we review  recent advances in the field. 

Much of the foundational work in selective inference relates to the problem of inference after variable selection in linear regression \citep{leeb2003finite, berk2013valid, lee2016exact, tibshirani2016exact, tian_selective_2018}. Recently, selective inference solutions have been developed for a much broader set of problems, such as inference after changepoint detection \citep{hyun2021post, jewell2022testing}, inference after outlier or anomaly detection \citep{chen2020valid, niihori2025statistically}, and inference after detecting ``regions" of interest  \citep{benjamini2019selection, nishino2025statistical}.  To illustrate the breadth of settings in which a need for selective inference arises, we frame this article around the following three motivating examples.

%\begin{example}[label=ex_winner]{Inference on a winner}
\begin{examplebox}[label=ex_winner]{Inference on a winner}
Suppose we observe realizations $Y_1,\ldots,Y_n$ corresponding to $n$ different \emph{candidates}, with $Y_k \overset{\mathrm{ind.}}{\sim} \N(\mu_k, \sigma^2)$. The parameter of interest is $\mu_{\hat{k}}$, where $\hat{k} = \argmax_k Y_k$. As  we are most likely to select an index $k$ precisely when $Y_k$ is unusually large, it follows that $Y_{\hat{k}}$ is biased for $\mu_{\hat{k}}$. Therefore, the classical $1-\alpha$ confidence interval for $\mu_k$ corresponding to some fixed index $k$, which takes the form
\begin{equation}
\label{eq:ex1-classical}
\sbr{Y_k \pm z_{1-\alpha/2} \sigma},
 \end{equation} 
does not attain $1-\alpha$ coverage when $k$ is replaced with $\hat{k}$. 
This issue is often referred to as the \emph{winner's curse}, and has been considered by a number of authors \citep{zhong2010correcting, fuentes2018confidence, andrews2024inference, zrnic2024locally, zrnic2025flexible}.
\end{examplebox}
%\end{example}

%\begin{marginnote}[]
%\entry{$z_{1-\alpha/2}$}
%\ Denotes the $1-\alpha/2$ quantile of the standard normal %distribution
%\end{marginnote}

%\begin{example}[label=ex_tree]{Inference on a regression tree}
\begin{examplebox}[label=ex_tree]{Inference on a regression tree}
Regression tree algorithms such as CART \citep{breiman1984classification}  partition  observations $\cbr{(X_i, Y_i)}_{i=1}^n$ into regions in covariate space with similar $Y_i$ values. We consider inference on $\mu_{\hat{R}} = \frac{1}{|i: X_i \in {\hat{R}}|} \sum_{i: X_i \in {\hat{R}}} \EE[Y_i]$, where ${\hat{R}}$ is one such region. Assuming that  $Y_i \overset{\mathrm{ind.}}{\sim}  \N(\mu_i, \sigma^2)$, the classical $1-\alpha$ confidence interval for $\mu_R$ corresponding to some fixed region $R$ takes the form
\begin{equation}
\label{eq:ex2-classical}
\sbr{\rbr{\frac{1}{|i: X_i \in R|} \sum_{i: X_i \in R} Y_i} \pm z_{1-\alpha/2} \sigma / \sqrt{|i: X_i \in R|}}.
 \end{equation} 
 This interval does not attain $1-\alpha$ coverage when $R$ is replaced with $\hat R$. This issue has been considered by \cite{athey2016recursive, wager2015adaptive, neufeld2022tree} and \cite{bakshi2024inference}. 
 %\end{example}
\end{examplebox}

%\begin{example}[label=ex_cluster]{Inference after clustering}
\begin{examplebox}[label=ex_cluster]{Inference after clustering}
The analysis of single-cell RNA sequencing data often centers around identifying genes that are differentially expressed across cell types. Since the cell types are not known \emph{a priori}, scientists (i) estimate the cell types by clustering the rows of $X$, a matrix of expression counts for $n$ cells and $p$ genes; and (ii) test the expression counts of the $j$th gene, $X_j$, for association with the estimated cluster labels. Assuming that $X_{ij} \overset{\mathrm{ind.}}{\sim} \N(\mu_{ij}, \sigma^2)$, 
the classical $1-\alpha$ confidence interval for a difference in means of feature $j$ across two pre-specified groups $A$ and $B$, denoted $\frac{1}{|i \in A|} \sum_{i \in A} \EE[X_{ij}] - \frac{1}{|i \in B|} \sum_{i \in B} \EE[X_{ij}]$, takes the form
\begin{equation}
\label{eq:ex3-classical}
%\CI_i^{1-\alpha}(Y)  := 
\sbr{\rbr{\frac{1}{|i \in A|} \sum_{i \in A} X_{ij} - \frac{1}{|i \in B|} \sum_{i \in B} X_{ij}} \pm z_{1-\alpha/2} \sigma\sqrt{1/|i \in A| + 1/|i \in B|}}.
\end{equation}
This interval does not attain $1-\alpha$ coverage when the pre-specified groups $A$ and $B$ are replaced with clusters estimated from the data. This issue has been described as one of the ``grand challenges" in single cell data science  \citep{lahnemann2020eleven}.
\end{examplebox}
%\end{example}

 In what follows, we highlight key considerations that arise in selective inference, as well as pros and cons of various selective inference frameworks, through the lens of these three examples. The rest of this article is organized as follows. In Section~\ref{sec_coverage}, we discuss conditional versus unconditional inferential guarantees, and argue using a simulation related to Example~\ref{ex_winner} that conditional guarantees are more  scientifically meaningful. We therefore focus the remainder of the review on methods that achieve conditional guarantees. In Sections~\ref{sec:methods} and \ref{sec:methods_2}, we review several distinct approaches to achieve conditional guarantees, and demonstrate that they can all be cast as special cases of a ``unifying recipe". A simulation study related to Example~\ref{ex_tree} is in Section~\ref{sec_simulation}, and  an application to single-cell RNA sequencing data related to Example~\ref{ex_cluster} is in 
 Section~\ref{sec_realdata}. We conclude with a discussion of challenges and future directions in Section~\ref{sec_discuss}. The Appendix contains details related to empirical results.

%% file: sections/sec2_new_coverage.tex
\section{Unconditional and conditional coverage}
\label{sec_coverage}

Suppose that we wish to construct a confidence interval around a parameter selected from the data, as in Examples~\ref{ex_winner}-\ref{ex_cluster}. \emph{What inferential guarantee should we demand of this interval?} 

To formalize this question, suppose that our data is distributed as $Y \sim F$. Let $\Kcal$ index a countable number of possible parameters $\theta_k(F)$ of interest to the analyst, and $\Theta := \cbr{\theta_k: k \in \Kcal}$, dropping dependence on $F$ for brevity.  We let $\Sel : \mathcal{Y} \rightarrow \Kcal$ denote a \emph{selection rule} that is applied to our data $Y \in \mathcal{Y}$ to obtain a data-driven parameter of interest $\theta_{\Sel(Y)}$. For instance, in Example~\ref{ex_winner}, $\Theta = \{\mu_1, \ldots, \mu_n\}$ and $\theta_{\Sel(Y)} = \mu_{\hat{k}}$, while in Example~\ref{ex_tree}, $\Theta = \{\mu_{R} :$ the region $ R \subset \mathbb{R}^p$  could result from  CART$\}$ and $\theta_{\Sel(Y)}=\mu_{\hat{R}}$. 

Let $\selCI$ be a $1-\alpha$ confidence interval based on our data $Y$ for the selected parameter $\theta_{\Sel(Y)}$. Perhaps surprisingly, there is not a consensus in the selective inference literature about what type of coverage this interval should achieve: proposals target either \emph{unconditional} or \emph{conditional} coverage. We will define these types of coverage in Sections~\ref{subsec:uncond} and \ref{subsec:cond}, and will then argue in Section~\ref{subsec:goal} in favor of a conditional coverage guarantee.  

\begin{remark}
The assumption of a countable number of possible parameters is made here for ease of exposition. As noted by \cite{fithian_optimal_2017}, these ideas generalize to an uncountable number of parameters with additional care. 
\end{remark}

\subsection{Unconditional coverage} \label{subsec:uncond}

Table~\ref{tab:classical-coverages} shows the proportion of simulated datasets for which the classical 90\% confidence intervals defined in 
Examples~\ref{ex_winner}-\ref{ex_cluster} 
contain the relevant selected parameter. Simulation details are provided in Appendix~\ref{appendix_naive}. This quantity is referred to as \emph{unconditional coverage}. 

\begin{definition}
\label{def:unconditional_coverage}
The interval $\selCI$ achieves the nominal \emph{unconditional coverage} if $\Pr\left( \theta_{{\Sel}(Y)} \in  \selCI \right) \geq 1-\alpha$.
\end{definition}

As shown in Table~\ref{tab:classical-coverages}, the unconditional coverage of the classical intervals falls well below 90\% when the signal is not strong.

\begin{table}[h]
\caption{In the context of Examples~\ref{ex_winner}-\ref{ex_cluster}, we  report the unconditional coverage of the classical 90\% confidence interval across 2000 simulated datasets for each of three levels of signal strength: none (corresponding to no signal), medium signal, and strong signal.  The classical interval does not account for the data-driven nature of the parameter, and thus severely under-covers the parameter when the signal strength is weak. Details of the simulation are provided in Appendix~\ref{appendix_naive}.} 
\centering
\small
\begin{tabular}{c|| c|c|c }
 & None & Medium & Strong  \\
\hline
\hline
Winner's Curse 
& 0.0085 & 0.779 & 0.888\\

Regression Tree 
& 0.358 & 0.887 & 0.910 \\
Clustering
& 0.370 & 0.858& 0.906\\
\hline
\hline
\end{tabular}
\label{tab:classical-coverages}
\end{table}

\begin{remark}
\label{remark:signal}
    In Table~\ref{tab:classical-coverages}, when the signal  is strong, the classical intervals have unconditional coverage very close to the nominal 90\%. This is because in the strong-signal setting, the selection event  captures the ``true" signal with  probability nearly one. For instance, in the context of the winner's curse in Example~\ref{ex_winner}, when the signal is strong, then $\Pr\left(\hat{k} = \argmax_k \mu_k\right) \approx 1$. In this case, the selection event is (nearly) deterministic, and no correction for selection is required. This is a common phenomenon in selective inference: failure to account for data-driven selection is most severe when the signal-to-noise ratio underlying the selection event is low. 
\end{remark}

In Definition~\ref{def:unconditional_coverage}, the parameter that we wish to cover is random. The two most common strategies to achieve unconditional coverage each involve a different strategy to circumvent the randomness in this parameter.
\begin{enumerate}
\item \emph{Seek simultaneous coverage of $\{\theta_k : k \in \mathcal{K}\}$, which is a deterministic set}. Such simultaneous coverage can be achieved by borrowing techniques from the multiple testing literature \citep{berk2013valid, bachoc2019valid, kuchibhotla2020valid}. That is, we first construct a set of intervals $\{\CI_{k}^{1-\alpha}(Y) : k \in \Kcal \}$ that satisfy $\Pr\left(  \cap_{k \in \Kcal} \left\{ \theta_{k} \in \CI_{k}^{1-\alpha}(Y)  \right\} \right) \geq 1-\alpha$. Of course, for all $k\in\Kcal$, $\Pr\left(    \theta_{k} \in \CI_{k}^{1-\alpha}(Y)  \right) \geq \Pr\left(  \cap_{k \in \Kcal} \left\{ \theta_{k} \in \CI_{k}^{1-\alpha}(Y)  \right\} \right)$. Therefore, for all $k\in\Kcal$, $\Pr\left(    \theta_{k} \in \CI_{k}^{1-\alpha}(Y)  \right) \geq 1-\alpha$. Thus, no matter the selection procedure $\Sel(Y)$, unconditional coverage follows directly \citep{berk2013valid}.
    \item \emph{Condition on the selection event, so that the selection event is no longer random.}
     This technique provides not only unconditional coverage but also a stronger guarantee, conditional coverage, which is the topic of the next subsection. 
\end{enumerate}
While some recent approaches for achieving unconditional guarantees do not precisely follow either of these strategies \citep{zrnic2024locally,andrews2024inference,zrnic2025flexible,zrnic2023post},  we do not provide details here, since --- as noted in the next section --- conditional guarantees will be the topic of the remainder of this article.

\begin{remark}\label{remark:uncond_pvalue}
What if we want a p-value for $H_0: \theta_{\Sel(Y)}=0$, rather than a confidence interval? It is not clear how to achieve this from the unconditional coverage guarantee that $\Pr\left( \theta_{{\Sel}(Y)} \in  \selCI \right) \geq 1-\alpha$, since that guarantee averages over different selections of $\Sel(Y)$.  Therefore, if p-values are the goal, then conditional coverage --- which we discuss next --- may be a more suitable alternative. 
\end{remark}

\subsection{Conditional coverage} \label{subsec:cond}

\begin{definition}
\label{def:conditional_coverage}
The interval $\selCI$ achieves the nominal \emph{conditional coverage}, also referred to in \cite{fithian_optimal_2017} as \emph{selective coverage}, if for any $k \in \Kcal$,
$
\Pr\left( \theta_{{\Sel}(Y)} \in \selCI \mid \Sel(Y) = k \right) \geq 1-\alpha.
$
\end{definition}

\begin{proposition}
 Conditional coverage  implies unconditional coverage. \label{prop:coverages}
\end{proposition}
\begin{proof}
Assume that $\selCI$ attains the nominal conditional coverage (Definition~\ref{def:conditional_coverage}). Then,
by the law of total probability, for a countable set of possible parameters $\mathcal{K}$,
\begin{align*}
\Pr\left( \theta_{{\Sel}(Y)} \in \selCI \right) &= \sum_{k \in \mathcal{K}} \Pr\left( \theta_{{\Sel}(Y)} \in \selCI \mid \Sel(Y) = k \right) \Pr\left(\Sel(Y) = k\right) \\
&\geq \sum_{k \in \mathcal{K}} (1-\alpha)\Pr\left(\Sel(Y) = k\right)  = (1-\alpha) \sum_{k \in \mathcal{K}} \Pr\left(\Sel(Y) = k\right) = (1-\alpha). 
\end{align*}
\end{proof}

It is natural to wonder whether techniques from multiple testing can be used to achieve conditional coverage. To answer this question, we revisit Example~\ref{ex_winner}.
%\begin{example}[label=ex:bonferroni]{Example 1, continued}
%\begin{examplebox}[label=ex:bonferroni]{Example 1, continued}

Suppose that we observe $Y_k \overset{\mathrm{ind.}}{\sim} \N(\mu_k, \sigma^2)$ for $k = 1,\ldots,n$, where $\sigma^2$ is known. 
Our interest lies in  $\mu_{\hat{k}}$, where $\hat{k} = \argmax_k Y_k$. 
Suppose we apply a  Bonferroni correction \citep{dunn1961multiple} over the $n$ candidates: that is, we define the interval $\CI_k^{1-\alpha}(Y) := \sbr{Y_{{k}} \pm z_{1-\frac{\alpha}{2n}} \sigma}$.
Following standard arguments, 
$\Pr\left( \cap_{i=k}^n \left\{ \mu_{k} \in \CI_{k}^{1-\alpha}(Y) \right\} \right) \geq 1- \alpha$. As pointed out in Section~\ref{subsec:uncond}, this implies that $\CI_{\hat k}^{1-\alpha}(Y)$, obtained by setting $k = \hat{k}$, attains nominal unconditional coverage in the sense of Definition~\ref{def:unconditional_coverage}. 

However, the Bonferroni-corrected intervals \emph{do not achieve conditional coverage}, in the sense of Definition~\ref{def:conditional_coverage}. To see this, we simulate data where  $\mu_1=4$, $\mu_2=\ldots=\mu_n=0$,  and $\sigma=1$. Thus, the first index is the ``true winner".  Table~\ref{tab:winner-bonferroni} displays Monte Carlo estimates (over 2,000 simulated datasets) of three quantities: the unconditional coverage, the coverage when the true winner is selected, and the coverage when the true winner is \emph{not} selected. First, we see that unconditional coverage holds as expected.
Furthermore, when the true winner is selected, the interval covers the selected parameter with very high probability. However, when the true winner is not selected, the probability of coverage is too low. Thus, \emph{conditional coverage does not hold.} 

What has happened? Recall that the first index is the ``true winner", and consider a fixed $k > 1$. We will have $\Sel(Y) =k$ precisely  in the subset of datasets for which $Y_k$ happens to take on a large value by chance; among such datasets, the Bonferroni interval is less likely to cover $\mu_k$. In other words, among datasets for which the wrong index is selected, the Bonferroni interval for the selected index tends to under-cover.
%\end{example}
  
\begin{table}[!htb]
\caption{%In Example~\ref{ex:bonferroni}, 
We assess the unconditional and conditional coverage of the Bonferroni-corrected 90\% intervals in a setting where the ``true winner" is the first index. While the Bonferroni intervals over-cover the parameter of interest unconditionally,  this is driven by over-coverage in the case where the ``correct" index is selected (i.e. $\Sel(Y)=1$). Among datasets where the correct index is not selected, the Bonferroni intervals do not attain nominal coverage, demonstrating that conditional coverage is not attained. Classical {90\%}  intervals are displayed for completeness. 
}
\centering
\small
\begin{tabular}{|c|c|| c|c| }
\hline
Type of Guarantee & Quantity & Bonferroni & Classical  \\
\hline
Unconditional & $
\Pr\left( \mu_{{\Sel}(Y)} \in \selCI \right)$
& 0.983 & 0.852 
\\
Conditional &  $
\Pr\left( \mu_{{\Sel}(Y)} \in \selCI \mid {\Sel(Y)} = 1 \right)$
& 1.00 & 0.936\\
Conditional &  $
\frac{1}{n-1} \sum_{k=2}^n \Pr\left( \mu_{{\Sel}(Y)} \in \selCI \mid {\Sel(Y)} = k \right)$
&0.811 & 0\\
\hline
\end{tabular}
\label{tab:winner-bonferroni}
\end{table}

Table~\ref{tab:winner-bonferroni} shows us that %Example~\ref{ex:bonferroni} has shown us that 
while simultaneous inference techniques such as the Bonferroni correction provide very strong guarantees --- much stronger than required for unconditional coverage! --- they do \emph{not} provide conditional coverage. 

\begin{remark}\label{remark:hypothesis}
We saw in Remark~\ref{remark:uncond_pvalue} that unconditional guarantees do not extend naturally to hypothesis testing. By contrast, constructing a hypothesis test from an interval that attains conditional coverage is straightforward: if the interval $CI_{\Sel(Y)}^{1-\alpha}(Y)$ attains $1-\alpha$ conditional coverage, then $\Pr_{H_0: \theta_{\Sel(y)}=0} \rbr{0 \notin \CI_{\Sel(Y)}^{1-\alpha}(Y) \mid \Sel(Y) = \Sel(y)} \leq \alpha$. This property is referred to as selective Type 1 error control by \cite{fithian_optimal_2017}.  
\end{remark}

\subsection{Our goal: conditional coverage} \label{subsec:goal}

\emph{Should we prefer conditional or unconditional coverage in practice?}

In this paper, we take the view of \citet{leeb_can_2006}, \citet{fithian_optimal_2017}, and \citet{kuffner2018principled}, who argue that conditional coverage (or its hypothesis testing analog; see Remark~\ref{remark:hypothesis}) ought to be the goal when carrying out inference on data-driven parameters. This is in contrast to authors who use techniques that attain conditional coverage  simply as a technique to achieve unconditional coverage (due to Proposition~\ref{prop:coverages}). 

To understand our view, consider Example~\ref{ex_winner}. In the motivating application considered by \cite{andrews2024inference}, $Y_k$ represents the estimated treatment effect of program $k$ among $n$ candidate programs in a randomized trial. Program $\hat{k} = \argmax_k Y_k$ is implemented as policy based on its promising performance. We wish to estimate $\mu_{\hat{k}}$, the expected treatment effect of the implemented program.  
Under the Bonferroni approach, we are overconfident in our inference on $\mu_{\hat{k}}$ precisely when we implemented a suboptimal program (i.e. when $\hat{k} \neq 1$; see Table~\ref{tab:winner-bonferroni}). In other words,  when we are \emph{wrong} (we do not select the true population winner), we are wrong \emph{twice} (we fail to cover the expected value of the selected winner).
On the other hand, a conditional approach would provide accurate inference for whichever policy we selected to deploy in practice. 

Therefore, throughout the rest of this paper, we focus on the goal of conditional coverage. 

\begin{remark} We acknowledge that there is room for  a difference in perspectives over whether conditional or unconditional coverage is the appropriate scientific goal. Of course, a reader whose primary interest lies in unconditional coverage can keep Proposition~\ref{prop:coverages} in mind as they read.  Moreover, in some cases the goal of conditional coverage is so stringent as to virtually eliminate all power (see the end of Section~\ref{subsec:full_cond} for a further discussion of this point); in such cases, a more modest goal of unconditional coverage may be more realistic.  
\end{remark}

%% file: sections/sec3_methods.tex
\section{Approaches for achieving conditional coverage}
\label{sec:methods}

In this section, we will see two approaches for obtaining conditional coverage, in the sense of Definition~\ref{def:conditional_coverage}: full conditional selective inference and sample splitting. While at first glance these approaches might appear quite different, we argue throughout this section that they are actually very similar: both involve conditioning on a selection event. In Section~\ref{sec:methods_2} we will show that, in fact, many additional strategies for achieving conditional coverage follow the same ``general recipe''.

\subsection{Full conditional selective inference}
\label{subsec:full_cond}

Consider a  \emph{fixed} parameter $\theta_k$ for some $k \in \Kcal$. Letting $F_{T_k(Y)}(\cdot, \theta_k)$ denote the CDF of a statistic $T_k(Y)$ for parameter $\theta_k$, it follows that 
\begin{equation}\label{fixed:ci}
   \CI^{1-\alpha}(T_k(Y)) := \cbr{\eta : F_{T_k(Y)}\rbr{T_k(Y); \eta} \in [\alpha/2, 1-\alpha/2]}
\end{equation}
is a valid ``classical" confidence interval for $\theta_k$ in the  sense that $\Pr\rbr{\theta_k \in \CI^{1-\alpha}(T_k(Y))} = 1-\alpha$. In fact, the classical confidence intervals in Equations \ref{eq:ex1-classical}-\ref{eq:ex3-classical} of Examples~\ref{ex_winner}--\ref{ex_cluster} are of this form. 

However, our goal is to construct a confidence interval for a data-driven parameter $\theta_{\Sel(Y)}$ with valid coverage conditional on $\cbr{\Sel(Y) = k}$, as defined in~\cref{def:conditional_coverage}. To achieve this goal, we will consider the  conditional CDF $F_{T_{\Sel(Y)}(Y) \mid \cbr{\Sel(Y) = k}}$, i.e., the distribution of the test statistic $T_{\Sel(Y)}(Y)$ conditional on the event $ \cbr{\Sel(Y) = k}$. The idea of characterizing this conditional CDF to account for inference on a data-driven parameter, either approximately or exactly, dates back to at least \citet{olshen_conditional_1973}, and has continued to be a focus of study in the years since then \citep{potscher_effects_1991, potscher_distribution_1998, leeb2003finite, leeb2005model, leeb_distribution_2005, leeb_can_2006, lee2016exact,tibshirani2016exact}. If this conditional CDF can be characterized, then a confidence interval for the selected parameter can be constructed as 
\begin{equation}\label{conditional:ci}
    \CI^{1-\alpha}(T_{\Sel(Y)}(Y) \mid \Sel(Y)) := \cbr{\eta : F_{T_{\Sel(Y)}(Y) \mid \Sel(Y)}\rbr{T_{\Sel(Y)}(Y); \eta} \in [\alpha/2, 1-\alpha/2]}.
\end{equation}
This interval is random in both $\Sel(Y)$ and $T_{\Sel(Y)}(Y)$. It follows from the probability integral transform that, conditional on any selection event $\Sel(Y)=k$,
\begin{equation}
    \Pr\rbr{\theta_{\Sel(Y)} \in \CI^{1-\alpha}(T_{\Sel(Y)}(Y)\mid \Sel(Y)) \mid \Sel(Y) = k }
    = 1-\alpha.
    \label{eq:cond_cov_proof}
\end{equation}

Intuitively, the confidence interval $\CI^{1-\alpha}(T_{\Sel(Y)}(Y)\mid \Sel(Y))$ provides conditional coverage because conditioning on all of the information used for selection means that none of this information is re-used for inference; thus, we do not ``double dip'' in that information.  

The selective inference literature often refers to inference using confidence intervals of the form \cref{conditional:ci} as  ``conditional selective inference". However, we will see in what follows that there exist other  strategies for obtaining  conditional (selective) coverage. Since \cref{conditional:ci}  uses the \emph{full} information in the data to select the parameter, we will refer to it as  a \emph{full conditional selective inference} interval, or ``full CSI" for short. 

In practice, it is often the case that the distribution of $T_{\Sel(Y)}(Y) \mid \cbr{\Sel(Y) = k}$ is intractable for inference, e.g., because it depends on unknown nuisance parameters. In such a setting, it may be fruitful to 
``condition on more''. Specifically, 
suppose that 
$T_{\Sel(Y)}(Y) \mid \cbr{\C(Y) = c}$ is tractable, where the event $\cbr{\C(Y) = c}$ informally ``implies the event'' $\cbr{\Sel(Y) = k}$ in the sense of \cref{prop:monotonicity}. Then it turns out that an interval that conditions on $\cbr{\C(Y)=c}$ will provide a conditional coverage guarantee in terms of $\cbr{\Sel(Y) = k}$.

\begin{proposition}\label{prop:monotonicity}
    Suppose that $\sigma \rbr{\C(Y)} \supseteq \sigma \rbr{\Sel(Y)}$, where $\sigma \rbr{\cdot}$ denotes the sigma-algebra generated by a random variable. Then, for
    \begin{equation}
    \label{eq:selci_more}
   \CI^{1-\alpha}(T_{\Sel(Y)}(Y) \mid \C(Y)) := \cbr{\eta : F_{T_{\Sel(Y)}(Y) \mid \C(Y)}\rbr{T_{\Sel(Y)}(Y); \eta} \in [\alpha/2, 1-\alpha/2]},
    \end{equation}
    we have that, for any $k \in \Kcal$,
    \begin{equation}
    \Pr\rbr{\theta_{\Sel(Y)} \in \CI^{1-\alpha}(T_{\Sel(Y)}(Y)\mid \C(Y)) \mid \Sel(Y) = k } = 1-\alpha.
    \end{equation}
\end{proposition}
\begin{proof}
    We apply the law of total expectation, the fact that $\sigma \rbr{\C(Y)} \supseteq \sigma \rbr{\Sel(Y)}$, and the probability integral transform: 
    \begin{align*}
    &\Pr\rbr{\theta_{\Sel(Y)} \in \CI^{1-\alpha}(T_{\Sel(Y)}(Y)\mid \C(Y)) \mid \Sel(Y) = k }\\
     &= \EE\Big[
     \Pr\Big(
     \theta_{\Sel(Y)} \in \CI^{1-\alpha}(T_{\Sel(Y)}(Y)\mid \C(Y)) \mid \C(Y), \Sel(Y)\Big) \bigm| \Sel(Y) = k 
     \Big]\\
    &= \EE\Big[
    \Pr\Big(
    \theta_{\Sel(Y)} \in \CI^{1-\alpha}(T_{\Sel(Y)}(Y)\mid \C(Y)) \mid \C(Y) \Big) \bigm| \Sel(Y) = k 
    \Big]\\
    &= \EE\left[1-\alpha\mid \Sel(Y) = k\right] = 1-\alpha.
    \end{align*}
\end{proof}

To make this concrete, we return to Example~\ref{ex_winner}.

%\begin{example}[label=ex:winners_more]{Example~\ref{ex_winner}, re-visited}
\begin{examplebox}[label=ex:winners_more]{Example~\ref{ex_winner}, re-visited}
% \begin{example}[Examples~\ref{ex_winner} and \ref{ex:bonferroni}, re-visited]\label{ex:winners_more}
We are interested in inference on the mean $\mu_{\Sel(Y)}$ where $\Sel(Y) := \argmax_{k \in \cbr{1,\dots,n}} Y_k$, where $Y_k \indsim  N(\mu_k, \sigma^2)$. At first glance it makes sense to consider the conditional distribution
\begin{equation}\label{eq:winner_v1}
    Y_{\Sel(Y)} \mid \{\Sel(Y) = k\} \sim \mathrm{TN}^{(\max_{k' \neq \Sel(Y)} Y_{k'}, \infty)}(\mu_{\Sel(Y)}, \sigma^2),
\end{equation}
where $ \mathrm{TN}^{(\max_{k' \neq \Sel(Y)} Y_{k'}, \infty)}(\mu_{\Sel(Y)}, \sigma^2)$ denotes a $\Norm(\mu_{\Sel(Y)}, \sigma^2)$ random variable truncated to the region $(\max_{k' \neq \Sel(Y)} Y_{k'}, \infty)$.
However, this lower bound of truncation is random, and depends on the unknown distribution of the second-largest observation. To remove the dependence on this unknown quantity, it suffices to condition on ``more'': the observed value of this lower bound, i.e., the event $\cbr{\max_{k' \neq \Sel(Y)} Y_{k'} = l}$. This yields the tractable truncated normal distribution
\begin{equation}\label{eq:winner_tn}
    Y_{\Sel(Y)} \Bigm| \cbr{\Sel(Y) = k, \max_{k' \neq \Sel(Y)} Y_{k'} = l} \sim \mathrm{TN}^{(l, \infty)}(\mu_{\Sel(Y)}, \sigma^2).
\end{equation}
Standard software can provide a confidence interval for the mean of a truncated normal random variable~\citep{EnvStats-book}. By the logic in Proposition~\ref{prop:monotonicity}, this confidence interval achieves valid conditional coverage~\citep{andrews2024inference}. 
%\end{example}
\end{examplebox}

In practice, carrying out full CSI requires the ability to characterize, or at least sample from, the conditional distribution of $T_{\Sel(Y)}(Y) \mid \C(Y) = c$. In some cases, it may be possible to apply a Monte Carlo strategy~\citep{fithian_optimal_2017}; however, this is not always applicable and is typically quite inefficient. In a  foundational piece of work, \cite{lee2016exact} showed that under multivariate normality it is possible to analytically characterize the conditional distribution of the test statistic for a coefficient in a linear regression model that was selected using the lasso.
Similar ideas were later developed for other selected parameters, such as those arising from  clustering~\citep{gao_selective_2022, chen_selective_2023, yun_selective_2023, chen2025testing}, regression trees~\citep{neufeld2022tree}, changepoints~\citep{hyun2021post, jewell2022testing, carrington_improving_2024}, principal components analysis~\citep{perry2025inference}, and online learning~\citep{chen_optimal_2023}. While preliminary work suggests that in some cases it may be possible to automate the task of characterizing the selection event~\citep{shiraishi_statistical_2025}, in general, developing an analytical characterization for a selection event is a time-consuming task that must be re-done each time a new selection event is considered. Furthermore, regardless of whether a Monte Carlo strategy or analytical characterization of the selection event is performed, full CSI requires that the selection event be formalized, i.e., a parameter selected by way of an  \emph{ad hoc} and unrecorded exploratory data analysis is unlikely to be a suitable candidate for full CSI.

A more nuanced drawback of full CSI is that conditioning on  $\C(Y)$ may leave little Fisher information about the parameter  for inference \citep{fithian_optimal_2017, kivaranovic_length_2021, andrews2024inference}. 
In the best case,  the distribution of the test statistic conditional on selection is ``close'' to its unconditional distribution. For example, in Example~\ref{ex:winners_more}, if $Y_k$ is the clear winner in the sense that $Y_k \gg \max_{k' \neq k} Y_{k'}$, then the confidence interval arising from the truncated normal will be similar to a confidence interval using a $\Norm(\mu_k, 1)$ distribution. However, if  $Y_k \approx \max_{k' \neq k} Y_{k'}$, then the truncated normal distribution will yield an \emph{extremely wide} confidence interval for $\mu_k$. In fact, \citet{kivaranovic_length_2021} prove that confidence intervals resulting from full CSI are infinitely wide in expectation. We return to this issue in \cref{subsec:sim_results}.

\subsection{Sample splitting}
\label{subsec:samp_split}

The use of sample splitting to avoid double dipping dates back to at least \cite{cox1975note}. We split a sample of independent observations $Y_1, \dots, Y_n$ into two disjoint subsets $\Ytrain$ and $\Ytest$, i.e., $\Ytrain \cup \Ytest = Y$ and $\Ytrain \cap \Ytest = \emptyset$. Selection is performed using $\Ytrain$, and inference for the parameter $\theta_{\Sel(\Ytrain)}$ is based on only the data $\Ytest$. 

What statistical guarantees does sample splitting provide? Because the subsets are  independent, for any selection function $\Sel(\cdot)$ and any test statistic $T_{\Sel(\cdot)}(\cdot)$, we have that $T_{\Sel\rbr{\Ytrain}}\rbr{\Ytest} \mid \cbr{\Ytrain = \ytrain} \equalindist T_{\Sel\rbr{\ytrain}}\rbr{\Ytest}$. Consequently, the classical confidence interval in~\cref{fixed:ci} coincides with the conditional confidence interval in \cref{eq:selci_more}: that is,
\begin{equation*}
 \CI^{1-\alpha}(T_{\Sel\rbr{\Ytrain}}(\Ytest) \mid \Ytrain)
 =
 \CI^{1-\alpha}(T_{\Sel\rbr{\Ytrain}}(\Ytest)).
\end{equation*}
Then, noting that $\sigma(\Ytrain) \supseteq \sigma(\Sel(\Ytrain))$, 
Proposition~\ref{prop:monotonicity} guarantees that
\begin{equation*}
 \Pr\rbr{\theta_{\Sel(\Ytrain)} \in \CI^{1-\alpha}(T_{\Sel(\Ytrain)}(\Ytest)) \mid \Sel(\Ytrain) = k}
= 1-\alpha.
\end{equation*}
This means that the very natural strategy of selecting a parameter based on $\Ytrain$ and constructing a classical interval based on $\Ytest$ provides conditional coverage. Furthermore, sample splitting has two key advantages over full CSI: 
\begin{enumerate}
    \item In sample splitting, \emph{no bespoke methods} are needed to construct the confidence interval, and 
 the selection event can be an \emph{ad-hoc}, adaptive exploratory data analysis. 
\item Sample splitting allocates all of the information in $\Ytrain$ solely for selection, and all of the information in $\Ytest$ solely for inference. By modifying the proportion of observations allocated to each set, the analyst has a high degree of control over this allocation (exact control in the setting where observations are identically distributed). By contrast, in full CSI, an unspecified amount of information is leftover for inference. 
 \end{enumerate}
However, a disadvantage of sample splitting is that the information in $\Ytrain$ not required for $\Sel(\Ytrain)$ is discarded; by contrast, in full CSI, no information is wasted. This is discussed further in \cref{subsec:recipe}. 

\begin{remark}\label{ref:samplesplit-winner}
While sample splitting is easy to apply, it is not always applicable! For instance, consider Example~\ref{ex_winner}, the winner's curse. If we applied sample splitting and selected, e.g., $\hat{k}=4$ using $\Ytrain$, then $\Ytest$ would contain no information about $\mu_{\hat k}$, since $\Ytrain \cap \Ytest = \emptyset$. More subtly, sample splitting is not applicable in the context of Example~\ref{ex_cluster}; we discuss this matter in Remark~\ref{remark:cluster_hard}.
\end{remark}

\begin{remark}
    One might think that the test set in sample splitting always contains information about the selected parameter, and thus, unlike full CSI, will not  yield infinitely wide confidence intervals. However, this is not always true, as we saw in Remark~\ref{ref:samplesplit-winner} and will further discuss  in \cref{sec_simulation}.
\end{remark}

\subsection{A unifying recipe}\label{subsec:recipe}

We now show that full CSI and sample splitting are two instances of a more general recipe for obtaining valid conditional coverage.

%\begin{textbox}[h]
\begin{mybox}[label=box:conditional]{Conditional coverage recipe}
%\boxlabel{box:conditional}
%\section{Box 1: Conditional coverage recipe}
    \begin{enumerate}
        \item The data is split into (possibly overlapping, or even identical) selection and inference sets.
        \item The target of inference is selected on the selection set. 
        \item Inference is conducted on the inference set, \emph{conditional on (at least) the event that this target was selected}.
    \end{enumerate}
%\end{textbox}
\end{mybox}

The recipe is outlined in simple language in Box~\ref{box:conditional}, and details are as follows. In Step 1, we construct a selection set $\Ytrain$ and an inference set $\Ytest$ from the full data $Y$; these sets contain a subset of the information in $Y$ but \emph{can be dependent}, and might even be equal to one another. In Step 2, a selection algorithm $\Sel(\cdot)$ selects the target of inference $\theta_{\Sel(\Ytrain)}$ using  $\Ytrain$. In Step 3, we construct a confidence interval $\CI^{1-\alpha}\left(T_{\Sel(\Ytrain)}(\Ytest) \mid \C(\Ytrain)\right)$ as in \cref{eq:selci_more}, where $\sigma \rbr{\C(\Ytrain)} \supseteq \sigma \rbr{\Sel(\Ytrain)}$. By Proposition~\ref{prop:monotonicity}, this interval attains the desired conditional coverage guarantee in \cref{def:conditional_coverage}.

How do sample splitting and full CSI relate to this recipe? In the context of full CSI, $\Ytrain = \Ytest = Y$, and $\sigma \rbr{\C(Y)} \supseteq \sigma \rbr{\Sel(Y)}$. Furthermore, $\sigma \rbr{\C(Y)} = \sigma \rbr{\Sel(Y)}$ if we do not ``condition on more."  In sample splitting, $\Ytrain \cup \Ytest = Y$ are disjoint sets, and $\sigma \rbr{\C(\Ytrain)} = \sigma \rbr{\Ytrain} \supseteq \sigma(\Sel(\Ytrain))$.

\subsection{The Fisher information trade-off}
\label{subsec:fisher}

We now consider how the recipe in Box~\ref{box:conditional} allocates information between selection and inference. 
This information tradeoff can be formalized through the lens of Fisher information, assuming that we are in a setting where the Fisher information $\Ical_{Y}(\theta)$ is well-defined for $Y \sim F$ parameterized by $\theta$.

First, we discuss the Fisher information available for \emph{selection} (Step 2 in Box~\ref{box:conditional}).  We note
that
$$\Ical_{\Sel(\Ytrain)} (\theta) \preceq \Ical_{\Ytrain} (\theta) \preceq \Ical_{Y} (\theta).$$

Thus, in a setting where $\Ytrain \subseteq Y$, the amount of information used for selection, $\Ical_{\Sel(\Ytrain)}(\theta)$, is bounded above by $\Ical_{\Ytrain}(\theta)$, the total information in $\Ytrain$, which is in turn at most $\Ical_{Y}(\theta)$, the information available for selection under full CSI. Intuitively, we might therefore expect that full CSI leads, in some sense, to a higher-quality selection event than sample splitting; we will see in \cref{sec_simulation} that this is typically true.

We now turn to the Fisher information available for \emph{inference}.
We use the notation $\Ical_{\Ytest \mid \C(\Ytrain) = \C(\ytrain)}(\theta)$
to represent the ``leftover Fisher information'' \citep{fithian_optimal_2017}; this is the Fisher information in the conditional distribution $\Ytest \mid \cbr{\C(\Ytrain) = \C(\ytrain)}$ with respect to a likelihood parameterized by $\theta$ (i.e., it is the information leftover from selection that is available for inference). Intuitively, the amount of leftover Fisher information should correspond to the width of the resulting confidence intervals (i.e., more leftover Fisher information will typically correspond to narrower intervals).  
In full CSI, the leftover Fisher information is $\Ical_{Y \mid \C(Y) = \C(y)}\rbr{\theta}$, while in sample splitting it is $\Ical_{\Ytest}\rbr{\theta}$. We note the following:
\begin{enumerate} 
    \item In full CSI, it is possible for $\Ical_{Y \mid \C(Y) = \C(y)}\rbr{\theta}$ to be arbitrarily close to zero (roughly speaking, this happens when we ``barely" selected $\Sel(y)$, i.e., $y$ is near the boundary of the selection event~\citep{fithian_optimal_2017}). In this case, full CSI will yield extremely-- or even infinitely-- wide confidence intervals. By contrast, when the conditional distribution of $Y \mid \cbr{\C(Y) = \C(y)}$  approximately equals the marginal distribution  of $Y$ (roughly speaking, when $y$ is deep in the interior of the selection event), full CSI yields $\Ical_{Y \mid \C(Y) = \C(y)}\rbr{\theta} \approx \Ical_{Y}\rbr{\theta}$. Thus, in this case, full CSI retains nearly all of the Fisher information about $\theta$ for inference, and the selective intervals are nearly identical to the classical intervals. 
    \item For sample splitting, in a setting where the observations are independent and identically distributed (i.i.d.) and a non-zero proportion $\epsilon$ of the observations are allocated to the $\Ytrain$, the leftover Fisher information $\Ical_{\Ytest}\rbr{\theta}=(1-\epsilon) \Ical_{Y}\rbr{\theta}$. While this means that the information available for inference in the i.i.d. setting is bounded away from $0$ and so intervals will never be infinitely wide, it also means that this information never approaches $\Ical_{Y}\rbr{\theta}$. 
\end{enumerate}

The information available for selection versus the information available for inference form two pieces of an  information tradeoff. For any parameter $\theta$, the total Fisher information available in $Y$ can be decomposed as
\begin{equation}\label{eq:fisher_decomp}
    \Ical_{Y}(\theta)
    =
    \Ical_{\Sel(Y)}(\theta) + \EE\sbr{\Ical_{\C(Y) \mid \Sel(Y)}(\theta)} + \EE\sbr{\Ical_{Y \mid \C(Y)}(\theta)}
    \succeq
    \Ical_{\Sel(Y)}(\theta) + \EE\sbr{\Ical_{Y \mid \C(Y)}(\theta)}. % + \EE\sbr{\Ical_{\C(Y) \mid \Sel(Y)}(\theta)}.
\end{equation}
Here, $\Ical_{\Sel(Y)}(\theta)$ is the Fisher information about $\theta$ in the data used for selection, and $\EE[\Ical_{Y \mid \C(Y)}(\theta)]$ is the average leftover Fisher information used for inference. This decomposition suggests an inherent tradeoff: the more  information is used for selection, the less average leftover Fisher information is available for inference~\citep{rasines_splitting_2023, neufeld_data_2024, dharamshi_generalized_2024, leiner2025data}. Moreover,  $\EE\sbr{\Ical_{\C(Y) \mid \Sel(Y)}(\theta)}$ represents the expected information lost due to the fact that we have ``conditioned on more"  in the sense of Proposition~\ref{prop:monotonicity}. 

\begin{remark}\label{remark:fisher_decomp}
In this section, we have discussed the Fisher information trade-off for the entire parameter $\theta$. However, for inference purposes, we really care about $\Ical_{Y \mid \C(\Ytrain)=\C(\ytrain)}(\theta_{\Sel(\ytrain)})$, as opposed to $\Ical_{Y \mid \C(\Ytrain)=\C(\ytrain)}(\theta)$. Fully characterizing an information trade-off involving the former may be a fruitful area for future work. 
\end{remark}

%% file: sections/sec4_methods_2.tex
\section{Further extensions of the general recipe}\label{sec:methods_2}

\cref{sec:methods} reviewed two of the best-known frameworks for achieving conditional coverage in the sense of Definition \ref{def:conditional_coverage}, full CSI and sample splitting, and showed that both follow the ``recipe'' outlined in Box~\ref{box:conditional}. In this section, we review  several additional proposals for constructing estimators that achieve conditional coverage, and show --- perhaps surprisingly --- that they all follow the same recipe. The details of $\Ytrain$, $\Ytest$, and $\sigma\rbr{\C(\Ytrain)}$ for full CSI, sample splitting, and the strategies discussed in this section are given in \cref{tab:conditional}.

\begin{table}[!htb]
    \centering
    \caption{Proposals for achieving conditional coverage that follow the recipe in Box~\ref{box:conditional}, summarized by  $\Ytrain$, $\Ytest$, and $\C(\Ytrain)$. Here $\zeta$ denotes an auxiliary random variable, and $g_{\text{sel}}$ and $g_{\text{inf}}$ denote specific functions of $Y$ and $\zeta$. A representative reference is provided for each proposal.}
    \label{tab:conditional}
    \begin{tabular}{ llllll } 
     \hline
     Framework & $\Ytrain$ & $\Ytest$ & $\sigma \rbr{\C(\Ytrain)}$ &  Note & Representative reference \\ 
     \hline
     Full CSI & $Y$ & $Y$ & $\supseteq \sigma \rbr{\Sel(\Ytrain)}$ & & \citet{lee2016exact} \\
     Sample splitting & $\subset Y$ & $Y \setminus \Ytrain$ & $= \sigma \rbr{\Ytrain}$ & $\Ytrain \independent \Ytest$ & \citet{cox1975note} \\
     Data carving & $\subset Y$ & $Y$ & $\supseteq \sigma \rbr{\Sel(\Ytrain)}$ & &\citet{fithian_optimal_2017}\\
     Data thinning & $g_{\text{sel}}(Y, \zeta)$ & $g_{\text{inf}}(Y, \zeta)$  & $= \sigma \rbr{\Ytrain}$ & $\Ytrain \independent \Ytest$ & \citet{dharamshi_generalized_2024} \\
    Randomized CSI & $g_{\text{sel}}(Y, \zeta)$ & $Y$ & $\supseteq \sigma \rbr{\Sel(\Ytrain)}$ & &\citet{tian_selective_2018} \\
     Data fission & $g_{\text{sel}}(Y, \zeta)$ & $g_{\text{inf}}(Y, \zeta)$ & $\supseteq \sigma \rbr{\Ytrain}$ & & \citet{leiner2025data} \\
     \hline
    \end{tabular}
\end{table}

\subsection{Data carving}

We saw in Section~\ref{subsec:fisher} that in sample splitting, none of the information in $\Ytrain$ is available for inference. 
Specifically, following from 
\eqref{eq:fisher_decomp}, sample splitting discards the information in $\EE\sbr{\Ical_{\Ytrain \mid \Sel(\Ytrain)}(\theta)}$. 

\citet{fithian_optimal_2017} propose \emph{data carving}, an  alternative to sample splitting that avoids discarding $\EE\sbr{\Ical_{\Ytrain \mid \Sel(\Ytrain)}(\theta)}$. In the language of  Box~\ref{box:conditional}, we set $\Ytrain \subset Y$ as in sample splitting, and set $\Ytest = Y$. Then inference is conducted using the conditional distribution of the test statistic $T_{\Sel(\Ytrain)}(Y) \mid \cbr{\Sel(\Ytrain) = \Sel(\ytrain)}$. Thus, the total Fisher information in $Y$ can decomposed as
\begin{equation}\label{eq:fisher_decomp_carving}
    \Ical_{Y}(\theta)
    =
    \Ical_{\Sel(\Ytrain)}(\theta) + \underbrace{\EE\sbr{\Ical_{Y \mid \Sel(\Ytrain)}(\theta)}}_{\text{Carving uses for inference}}
    =
    \Ical_{\Sel(\Ytrain)}(\theta) + \underbrace{\EE\sbr{\Ical_{\Ytrain \mid \Sel(\Ytrain)}(\theta)}}_{\text{SS discards}} + \underbrace{\Ical_{Y \setminus \Ytrain}(\theta)}_{\text{SS uses for inference }}.
\end{equation}

We close with a comparison of data carving to sample splitting and full CSI:
\begin{enumerate}
    \item Relative to sample splitting:
    \begin{enumerate}
        \item Data carving uses the same information for selection. 
        \item As mentioned earlier, data carving uses more Fisher information for inference because $\EE\left[\Ical_{Y \mid \Sel(\Ytrain)}(\theta)\right] \succeq \Ical_{Y \setminus \Ytrain}(\theta)$ in \eqref{eq:fisher_decomp_carving} In fact, \cite{fithian_optimal_2017} show that it is strictly more powerful. 
        \item Out-of-box tools cannot be used for inference, because $\Ytrain$ and $\Ytest$ are not independent. 
    \end{enumerate}
    \item Relative to full CSI:
    \begin{enumerate}
    \item Data carving allocates less information for selection, and so it typically will have lower ``quality" of selection.
 \item  The Fisher information $\Ical_{\Ytrain \mid \Sel(\Ytrain) = \Sel(\ytrain)}(\theta)$ available for inference in data carving is bounded above zero when observations are i.i.d.
 \item The distribution of $T_{\Sel(\ytrain)}(Y) \mid \{\Sel(\Ytrain) = \Sel(\ytrain)\}$ may be challenging to work with, as is often the case for full CSI. 
 \end{enumerate}
\end{enumerate}
2(c) above highlights that data carving can be computationally challenging. While some challenges may be alleviated by instead conducting inference using the conditional distribution of  $T_{\Sel(\ytrain)}(Y) \mid \rbr{\C(\Ytrain) = \C(\ytrain)}$ for a carefully-chosen $\C(\cdot)$ (as in Proposition~\ref{prop:monotonicity}), in general software for data carving is not readily available. 

\subsection{Data thinning}\label{subsec:data_thin}

In Section~\ref{subsec:samp_split}, we saw that sample splitting decomposes the data into two independent pieces by partitioning the $n$ observations. Is there a different way to decompose a dataset into independent pieces?

To make this idea concrete, suppose $Y \sim \Norm(\mu,\sigma^2)$ where $\sigma$ is known. For a user-specified constant $\epsilon \in (0, 1)$, draw $\zeta \sim \Norm(0, \epsilon(1-\epsilon)\sigma^2)$, and let $\Ytrain = \epsilon Y + \zeta$ and $\Ytest = Y - \Ytrain$. Then one can show that $\Ytrain \sim \Norm(\epsilon \mu, \epsilon \sigma^2)$, $\Ytest \sim\Norm((1-\epsilon)\mu, (1-\epsilon)\sigma^2)$, and $\Ytrain$ and $\Ytest$ are independent~\citep{robins2006adaptive, rasines_splitting_2023}. The choice of $\epsilon$ is the fraction of the total Fisher information allocated to selection, with the remaining $(1-\epsilon)$ allocated for inference~\citep{neufeld_data_2024}. We illustrate in Example~\ref{ex:winners_datathin} how this is useful in the context of the winner's curse~\citep{ma_breaking_2023}. A similar decomposition of a Poisson random variable $Y \sim \text{Poisson}(\mu)$ into independent components is also a well-known result~\citep{casella2002statistical}.

%\begin{example}[label=ex:winners_datathin]{Data thinning for winner's curse}
\begin{examplebox}[label=ex:winners_datathin]{Data thinning for winner's curse}
In the setting of Example~\ref{ex_winner}, suppose that $Y_i \sim N(\mu_i, \sigma^2)$. We apply the decomposition described above to each $Y_i$, $i=1,\ldots,n$ to obtain independent $\Ytrain$ and $\Ytest$ that each have dimension $n$. We then select $\Sel(\Ytrain) = \argmax_{k \in [1, \dots, n]} \Ytrain_k$, and conduct inference using the distribution of the test statistic $\Ytest_k \mid \cbr{\Sel(\Ytrain) = k}$, which thanks to independence is simply $\Ytest_k \sim \Norm((1-\epsilon)\mu_k, (1-\epsilon)\sigma^2)$.
%\end{example}
\end{examplebox}

\cite{neufeld_data_2024} and \cite{dharamshi_generalized_2024} showed that the ability to decompose a random variable $Y$ into two independent and non-trivial components extends far beyond the special cases of Gaussian and Poisson distributions. Their \emph{data thinning} framework enables  $Y$ to be ``thinned" (decomposed) into two independent components which are sufficient to reconstruct $Y$, provided that $Y \sim F(\theta)$ for a distributional family $F(\cdot)$ that is known up to an unknown parameter $\theta$. Distributional families that can be thinned in this way include binomial, multinomial, Gamma, exponential, truncated uniform, and negative binomial (with known overdispersion), among others.

How does this relate to Box~\ref{box:conditional} and Table~\ref{tab:conditional}?
In the language of Table~\ref{tab:conditional},  $\Ytrain$ and $\Ytest$ are \emph{independent} random variables generated by applying the functions 
$g_{\text{sel}}(\cdot)$ and $g_{\text{inf}}(\cdot)$ to 
$Y$ and an auxiliary random variable $\zeta$. For example, in the case of $Y \sim N(\mu,\sigma^2)$, we take $\zeta \sim \Norm(0, \epsilon (1-\epsilon)\sigma^2)$, $g_{\text{sel}}(Y, \zeta)=\epsilon Y+\zeta$, and $g_{\text{inf}}(Y, \zeta)= (1-\epsilon) Y - \zeta$. By virtue of independence between $\Ytrain$ and $\Ytest$, selection can be conducted using $\Ytrain$ and inference using $\Ytest$.
Further details regarding the distribution of $\zeta$ and the functions $g_{\text{sel}}(\cdot)$ and $g_{\text{inf}}(\cdot)$ for specific distributional families are provided in \cite{dharamshi_generalized_2024}.

At a high level, data thinning is quite similar to sample splitting.
\begin{enumerate}
\item Both approaches split the data into two independent pieces. This enables  inference on the test set using out-of-the-box tools. 
\item Both methods involve a tuning parameter that determines the amount of Fisher information in the training set versus the test set~\citep{neufeld_data_2024}.
\end{enumerate}

Yet data thinning differs from sample splitting in notable ways:
\begin{enumerate}
    \item We see in Example~\ref{ex:winners_datathin} that data thinning can be applied to the winner's curse, whereas sample splitting cannot (see Remark~\ref{ref:samplesplit-winner}). Other settings where data thinning can be applied but sample splitting cannot include inference on an estimated   network~\citep{chen_estimating_2025, ancell_post-selection_2026} and inference after clustering~\citep{neufeld2024inference}. 
    \item If the samples are not identically distributed, then data thinning may lead to narrower confidence intervals than sample splitting, even when the information allocated for selection is identical. This follows from an application of Jensen's inequality \citep{rasines_sampling_2025,neufeld_data_2024}.
    \item The application of data thinning is limited to certain families of distributions, and nuisance parameters may preclude its use. In some settings where data thinning cannot be directly applied, asymptotic Gaussianity of a suitable statistic may enable asymptotically valid inference ~\citep{rasines_splitting_2023, lei_discussion_2025, perry_post-selection_2026}.
\end{enumerate}

\subsection{Randomized conditional selective inference}\label{subsec:rcsi}

We saw in \eqref{eq:fisher_decomp_carving} that sample splitting leaves information on the table: in particular, it discards  $\EE\sbr{\Ical_{\Ytrain \mid \Sel(\Ytrain)}(\theta)}$. Data carving exploited this information by considering the conditional distribution of $T_{\Sel(\Ytrain)}(Y) \mid \cbr{\Sel(\Ytrain) = \Sel(\ytrain)}$ instead of $T_{\Sel(\Ytrain)}(\Ytest)$ where $\Ytrain \independent \Ytest$.
Like sample splitting, data thinning discards $\EE\sbr{\Ical_{\Ytrain \mid \Sel(\Ytrain)}}$. In this section, we present \emph{randomized CSI}, which makes use of the information left unused by data thinning. 

As was the case for data thinning, for selection, randomized CSI uses $\Ytrain=g_{\text{sel}}(Y, \zeta)$ for some user-generated auxiliary random variable $\zeta$. For inference, however, randomized CSI uses \emph{all} of the data, i.e., $\Ytest = Y$. Since $\Ytrain$ and $\Ytest$ are not independent, inference is conducted using the conditional distribution $T_{\Sel(\Ytrain)}(\Ytest) \mid \left\{ \Sel(\Ytrain)=\Sel(\ytrain) \right\}$ (or, possibly,   $T_{\Sel(\Ytrain)}(\Ytest)\mid \left\{ \C(\Ytrain)=\C(\ytrain) \right\}$ where $\sigma \rbr{\C(\Ytrain)} \supseteq \sigma \rbr{\Sel(\Ytrain)}$; see Proposition~\ref{prop:monotonicity}). 

The Fisher information in randomized CSI can be decomposed and related to that of data thinning (assuming they make use of the same $g_{\text{sel}}(Y, \zeta)$) as follows:
\begin{align*}
    \Ical_{Y}(\theta)
    &=
    \Ical_{\Sel(g_{\text{sel}}(Y, \zeta))}(\theta)
    + \underbrace{\EE\sbr{\Ical_{Y\mid \Sel(g_{\text{sel}}(Y, \zeta))}(\theta)}}_{\text{RCSI uses for inference}}\\
    &= \Ical_{\Sel(g_{\text{sel}}(Y, \zeta))}(\theta)
    + 
    \underbrace{\EE\sbr{
    \Ical_{g_{\text{sel}}(Y, \zeta) \mid \Sel(g_{\text{sel}}(Y, \zeta))}(\theta)
    }}_{\text{Thinning discards}} + \underbrace{\EE\sbr{\Ical_{g_{\text{inf}}(Y, \zeta)}(\theta)}}_{\text{Thinning uses for inference}}.
\end{align*}

Randomized CSI  was first proposed by \cite{tian_selective_2018}, and has since been instantiated in a variety of settings including penalized (generalized) linear models~\citep{tian_selective_2018, tian_magic_2016, panigrahi_integrative_2021, panigrahi2023approximate, panigrahi_carving_2023, huang_selective_2024}, randomized regression trees~\citep{bakshi2024inference}, and maximal contrasts~\citep{perry2024infer}. Randomized CSI also arises from starting with a data thinning estimator and Rao-Blackwellizing it \citep{ma_breaking_2023}.

We close this section with the following notes:
\begin{enumerate}
\item Randomized CSI avoids the infinite width confidence intervals that can arise from full CSI. While one might think that this comes at the cost of degraded ``selection quality", it turns out that a very small amount of randomization suffices to drastically reduce confidence interval widths relative to full CSI ~\citep{tian_selective_2018}.
\item Compared to data thinning, the cost of more information available for inference is much less flexibility. While inference in data thinning only requires knowledge of the distribution of $\Ytest$, randomized CSI requires the derivation of a new conditional distribution for each selection rule of interest.
\end{enumerate}

\subsection{Data fission}\label{subsec:fission}

\cite{leiner2025data} propose \emph{data fission}, a framework for decomposing  a random variable $Y$ into two pieces $\Ytrain$ and $\Ytest$ such that (i) both pieces contain information about the unknown parameter of interest, and (ii) both the marginal distribution of $\Ytrain$ and the conditional distribution of $\Ytest \mid \Ytrain$ are tractable. In the language of Table~\ref{tab:conditional}, data fission sets $\Ytrain$ and $\Ytest$ to be random variables generated from $Y$ and a user-generated auxiliary random variable $\zeta$ via functions $g_{\text{sel}}$ and $g_{\text{inf}}$. In the special case where $\Ytrain$ and $\Ytest$ are independent, the framework corresponds exactly to data thinning, described in Section~\ref{subsec:data_thin}. In this paper, we use the term ``data fission" to refer specifically to the case where $\Ytrain$ and $\Ytest$ are not independent. 

\cite{leiner2025data} suggest performing selection using $\Ytrain$ and using the conditional distribution of $\Ytest \mid \Ytrain$ for inference, meaning that $\sigma\rbr{\C(\Ytrain)} = \sigma\rbr{\Ytrain}$. Provided that $g_{\text{sel}}(Y, \zeta)$ and $g_{\text{inf}}(Y, \zeta)$ are sufficient to reconstruct $Y$, we can relate the Fisher information in data fission to that in randomized CSI:
\begin{align*}
    \Ical_{Y}(\theta)
    &=
    \Ical_{\Sel(g_{\text{sel}}(Y, \zeta))}(\theta)
    + \underbrace{\EE\sbr{\Ical_{Y\mid \Sel(g_{\text{sel}}(Y, \zeta))}(\theta)}}_{\text{RCSI uses for inference}}\\
     &= \Ical_{\Sel(g_{\text{sel}}(Y, \zeta))}(\theta)
     + 
     \EE\sbr{
     \Ical_{g_{\text{sel}}(Y, \zeta),g_{\text{inf}}(Y, \zeta) \mid \Sel(g_{\text{sel}}(Y, \zeta))}(\theta)
     } \\
    &= \Ical_{\Sel(g_{\text{sel}}(Y, \zeta))}(\theta)
    + 
    \underbrace{\EE\sbr{
    \Ical_{g_{\text{sel}}(Y, \zeta) \mid \Sel(g_{\text{sel}}(Y, \zeta))}(\theta)
    }}_{\text{Fission discards}} + \underbrace{\EE\sbr{\Ical_{g_{\text{inf}}(Y, \zeta) \mid g_{\text{sel}}(Y, \zeta)}(\theta)}}_{\text{Fission uses for inference}}.
\end{align*}
Compared to randomized CSI, this can be viewed as \emph{conditioning on more}, as $\sigma\rbr{\Ytrain} \supseteq \sigma\rbr{\Sel(\Ytrain)}$. The benefit of this approach is that, while randomized CSI requires a conditional distribution to be derived for \emph{each selection event of interest}, data fission conditions on all of $\Ytrain$, and thus conditional inference is valid no matter the form of $\Sel(\cdot)$. 

Both data fission and data thinning use $\sigma\rbr{\C(\Ytrain)} = \sigma\rbr{\Ytrain}$, but since the former does not require independence of $\Ytrain$ and $\Ytest$, it has the potential to be more broadly applicable than data thinning. For instance, it can be applied in settings where data thinning is misspecified \citep{neufeld2025discussion,dharamshi2025decomposing}. 
It is also possible to apply fission in certain settings where thinning is impossible. For example, it is impossible to thin a Bernoulli random variable~\citep{dharamshi_generalized_2024}, but fission has recently been applied to binary data in the contexts of penalized logistic regression~\citep{leiner2025data, neufeld2025discussion} and networks~\citep{ancell_post-selection_2026}. We highlight a few key features of data fission below. 
\begin{enumerate}
\item Data fission requires  inference using the conditional distribution of $\Ytest \mid \Ytrain$. Inference on the parameter of interest in this conditional distribution can be surprisingly difficult \citep{dharamshi2025decomposing, neufeld2025discussion, ancell_post-selection_2026}, and standard statistical software cannot be used. In some cases, we can still further ``condition on more" to make inference tractable \citep{perry2024infer} (see Proposition~\ref{prop:monotonicity}). It is worth noting that the practical challenges that arise in applying data fission are quite different from those that arise with full or randomized CSI. In the latter, challenge arises due to the need to characterize $T_{\Sel(\Ytrain)}(\Ytest)\mid\cbr{\Sel(\Ytrain) = \Sel(\ytrain)}$. By contrast, in the former, $\Ytest \mid \Ytrain$ is tractable by design: however, the parameter of interest $\theta_{\Sel(\Ytrain)}$ is often entangled in this conditional distribution in such a way that inference is extremely challenging.
\item While the conditional distribution of $\Ytest \mid \Ytrain$ can be challenging to work with, it does not need to be separately specified for every selection event $\Sel(\Ytrain)$ (in contrast to full or randomized CSI, for example). Thus, innovations in inference after fission in one setting may be easily transferrable to new settings with different selection events, mirroring the flexibility of data thinning and sample splitting. 
\item Unlike data thinning and sample splitting, the user 
does not have fine-grained control over the amount of Fisher information allocated to $\Ytrain$ versus $\Ytest \mid \Ytrain$ \citep{dharamshi2025decomposing,neufeld2025discussion}.
\end{enumerate}

%% file: sections/sec5_sim.tex
\section{Simulation Study}
\label{sec_simulation}

We now revisit Example~\ref{ex_tree}.  Our goal is to conduct inference on the mean of a region estimated using a regression tree algorithm, such as CART \citep{breiman1984classification}. 

Why might such inference be of interest? As one example, \cite{loh2019subgroup} use regression trees to identify subgroups of patients that  share a common mean value of a response variable. Valid confidence intervals for the population mean within each patient subgroup could be used to ascertain whether the observed differences between estimated subgroups are noise artifacts or contain true signal.

We saw in Table~\ref{tab:classical-coverages} that conducting inference on the mean within each subgroup using the classical interval in \eqref{eq:ex2-classical} will not achieve the nominal coverage, since the subgroups themselves (and, therefore, the mean parameter within each subgroup) are  functions of the data.  
In what follows, we explore methods that yield confidence intervals that achieve valid conditional coverage, in the sense of \eqref{conditional:ci}, where the selection event is the event that a particular region of covariate space was identified by a regression tree algorithm. 

\subsection{Simulation setup and inferential goals}
\label{subsec_setup}

We generate $(X_i, Y_i)$ for $i=1,\ldots,n$, with $X_i \overset{\mathrm{ind.}}\sim N_p(0, I_p)$ and $Y_i  \overset{\mathrm{ind.}}\sim N(\mu_i, \sigma_y^2)$. We treat $X=x$ as fixed, and let $\mu \in \mathbb{R}^n$ be a piecewise constant model over
regions of covariate space. Specifically, along the lines of \cite{neufeld2022tree} and \cite{bakshi2024inference}, we let 
\begin{equation}
\label{eq_truetree}
\mu_i = \beta \cdot \boldsymbol{1}(x_{i1} < 0) + \beta  \cdot \boldsymbol{1}(x_{i1} < 0 \ \& \ x_{i2} > 0) + \beta \cdot \boldsymbol{1}(x_{i1} > 0 \ \& \ x_{i3} > 1).
\end{equation}
Throughout our simulations, we fix $n=200,p=20,$ and $\sigma_y = 2.5$. We generate 1000 datasets in each of three signal settings: ``weak", ``medium", and ``strong", which correspond to $\beta=1, 2, 5$, respectively. 

\begin{mybox}[label=box:tree]{General procedure for inference on a two-level regression tree}
%\begin{textbox}[h]
%\boxlabel{box:tree}
%\section{Box 2: General procedure for inference on a two-level regression tree}
    \begin{enumerate}
        \item Fit a two-level regression tree to construct regions $\hat{R}_k$ for $k=1,\ldots,4$ such that (i) $\hat{R}_k \subset \mathbb{R}^p$, (ii) $\hat{R}_k \cap \hat{R}_{k'} = \emptyset$, and (iii) $\bigcup_{k=1}^4 \hat{R}_k= \mathbb{R}^p$. 
        \item Construct a $1-\alpha$ confidence interval for $\mu_{\hat{R}_k} = \frac{1}{|x_i  \in {\hat{R}_k}|} \sum_{x_i \in {\hat{R}_k}} \mu_i$, for $k=1,\ldots,4$. 
 \end{enumerate}
%\end{textbox}
\end{mybox}

%We consider the procedure outlined in Box~\ref{box:tree}.
In Section~\ref{subsec_compmethods}, we apply ideas from Sections~\ref{sec:methods} and \ref{sec:methods_2}  to the procedure in Box~\ref{box:tree}. 

\subsection{Methods for Comparison}
\label{subsec_compmethods}

We compare the following instantiations of Box~\ref{box:tree}. %for carrying out our task. %For simplicity, we assume throughout this section that $\sigma_y$ is known. 

\begin{list}{}{}
\item[\textbf{Classical method}:] We use all of the data both to fit a CART regression tree  in Step 1  of Box~\ref{box:tree}, and to construct classical Z-intervals~\cref{eq:ex2-classical} in Step 2 of Box ~\ref{box:tree}. 
%, we let $\xtrain = \xtest = x$, $\Ytrain=\Ytest=Y$, and construct classical Z-intervals (that ignore selection). %(using the known value of $\sigma_y$) for $\mu_{\hat{R}_k}$ for $k=1,\ldots,4$. %In other words, we apply Box~\ref{box:tree} with $\Ytrain = \Ytest = Y$, and classical Z-intervals.  
(We already know from Table~\ref{tab:classical-coverages} that this method does not achieve nominal coverage, but we include it for the sake of comparison.) 

\item[\textbf{Sample splitting:}] For $\epsilon \in \{0.5, 0.75, 0.9\}$, we use  $\epsilon n$ observations to fit a CART tree in Step 1 of  Box~\ref{box:tree}, and use the remaining $(1-\epsilon)n$ observations   to construct a classical Z-interval~\cref{eq:ex2-classical} for $\mu_{\hat{R}_1},\ldots,\mu_{\hat{R}_4}$ defined in Step 2 of Box~\ref{box:tree}. 
If no test observations are contained in $\hat{R}_k$, then the test observations contain no information about $\mu_{\hat{R}_k}$, and the corresponding confidence interval is $[-\infty, \infty]$. 
\item[\textbf{Data thinning:}] For $\epsilon \in \{0.5, 0.75, 0.9\}$, and for $i=1,\ldots,n$,  we apply Gaussian data thinning to decompose $Y_i$ into \emph{independent} components $\Ytrain_i$ and $\Ytest_i$,   which contain proportions $\epsilon$ and $1-\epsilon$ of the Fisher information about $\mu_i$. (See \cite{neufeld_data_2024} for details about Fisher information for data thinning.)
We then fit a CART tree on $(X,\Ytrain)$ in Step 1 of Box~\ref{box:tree},  and in Step 2  construct a classical Z-interval~\cref{eq:ex2-classical} for $\mu_{\hat{R}_k}$ using $(X, (1-\epsilon)^{-1} \Ytest)$; this scaling is required because $\E(\Ytest) = (1-\epsilon) \E(Y)$. 
Since the same $X$ values are used in Steps 1 and 2, the test set always contains information about  $\mu_{\hat{R}_k}$, and thus all confidence intervals have finite width.
\item[\textbf{Full conditional selective inference:}] 
We use all of the data  to fit a CART tree  in Step 1  of Box~\ref{box:tree}. We then conduct inference in Step 2 of Box ~\ref{box:tree} conditional on the event that the regions $\hat{R_1},\ldots,\hat{R}_4$ were selected from the data, using the proposal of \cite{neufeld2022tree}. 
\item[\textbf{Randomized conditional selective inference:}] We apply the proposal of \cite{bakshi2024inference}, which conducts Step 1 in  Box~\ref{box:tree} by applying a ``randomized" version of the CART  algorithm to the full data. This is an instantiation of randomized CSI, but instead of applying the usual deterministic CART algorithm to $(X, \Ytrain)$ where $\Ytrain$ is a noisy version of $Y$, we instead apply a randomized version of  the CART algorithm to $(X,Y)$. 
This randomized version  involves the injection of noise with variance $\tau^2$; we consider $\tau^2 \in \{0.25 \sigma_y^2, \sigma_y^2, 4 \sigma_y^2\}$. Intuitively, a larger value of $\tau^2$ corresponds to using less Fisher information  in Step 1, and thus reserving more  for Step 2. 
In Step 2 of Box~\ref{box:tree}, a  Z-interval is constructed for $\mu_{\hat{R}_k}$ using the distribution of $Y$ conditional on the event that the randomized CART algorithm selected the regions $\hat{R}_1,\ldots,\hat{R}_4$. 
\end{list}

\begin{remark}
While the classical method and sample splitting can seamlessly accommodate the setting where $\sigma_y^2$ is unknown, theoretical guarantees for the other methods require $\sigma_y^2$ to be known. In our simulations, we apply all methods using the true value of $\sigma_y^2$. 
\end{remark}

\begin{remark} A data fission approach to Box~\ref{box:tree}
could avoid the need to know $\sigma_y^2$, but the literature does not yet contain the details for such an approach. 
\end{remark}

\subsection{Metrics for Comparison}
\label{subsec_metrics}

We compare the five methods introduced in Section~\ref{subsec_compmethods} according to the following metrics. 
\begin{list}{}{}
\item \textbf{Selection quality}: Each  two-level tree that we fit using CART or randomized CART partitions the $n$ observations into four groups (corresponding to the regions $\hat{R}_1,\ldots,\hat{R}_4$). The true model  underlying the data \cref{eq_truetree} likewise assigns each of the $n$ observations to one of four groups. The adjusted Rand index \citep{hubert1985comparing} between these estimated and true assignment vectors captures whether the ``true" regions were accurately selected.  
\item \textbf{Coverage:} For every selected region $\hat{R}_k$ in every estimated tree, we check whether the corresponding confidence interval contains the corresponding parameter $\mu_{\hat{R}_k}$. 
\item \textbf{Confidence interval length:} We record the lengths of the confidence intervals.  
\end{list}

\subsection{Results}\label{subsec:sim_results}

The results of our simulation study, with $\alpha=0.1$ in Step 2 of Box~\ref{box:tree}, are in Figure~\ref{fig_simresults}.

First of all, as expected, all methods besides the classical method attain the nominal coverage of $90\%$. Furthermore, as discussed in Section~\ref{sec_coverage}, even the classical method attains nominal coverage when enough signal is present in the data (i.e., when $\beta=5$ in~\cref{eq_truetree}). 

Second, the selection quality of all methods increases with signal strength, as expected. Full CSI has the same selection quality as the classical method, since their selection events are identical. By contrast, for sample splitting, data thinning, and randomized CSI, the selection quality decreases as less information is allocated for selection (i.e. as $\epsilon$ decreases for sample splitting and data thinning, or as $\tau^2$ increases for randomized CSI).

In the following subsections, we discuss the trade-off between selection quality and confidence interval length for each method that attains nominal coverage (i.e., for all but the classical method).

\subsubsection{Sample splitting}

As $\epsilon$ --- the proportion of observations used for selection  --- decreases, both the quality of selection and the length of the  confidence intervals decreases. 
When the signal is high, it may be the case that not all of the information in the training set is required for selection; however, since sample splitting discards the information in the training set that is not used for selection (see Section~\ref{sec:methods}), 
the length of the resulting intervals does not vary with the amount of signal in the data.

As mentioned in Section~\ref{subsec_compmethods}, when we conduct sample splitting, it is possible that a region selected in Step 1 of Box~\ref{box:tree} may contain no test set observations. In that case, there is no information about the selected parameter in the test set, and so Step 2 of Box~\ref{box:tree} will yield a  confidence interval  of $[-\infty, \infty]$. For visualization purposes, such intervals are plotted as having length 10,000 in the right panel of Figure~\ref{fig_simresults}.  
 
The presence of infinite-width confidence intervals reflects a more general property of sample splitting: while  sample splitting will use \emph{exactly} a fraction $1-\epsilon$ of the Fisher information in the data for inference when the observations are i.i.d., when the observations are not i.i.d. (as in this example, because $X$ is fixed) the proportion used for inference is  $1-\epsilon$ \emph{in expectation}. In any given data realization, it is possible that no Fisher information at all will be left over for inference!  

\subsubsection{Data thinning}

Data thinning yields similar selection quality to sample splitting, but with the pronounced advantage that  the expected length of confidence intervals is finite. This is because --- unlike sample splitting --- data thinning  places \emph{exactly} $(1-\epsilon)$ of the information about any unknown parameter in the inference set \citep{neufeld_data_2024}. Of course, data thinning also suffers from a major disadvantage relative to sample splitting: whereas the latter requires only independent observations, the former requires that  $Y_i \overset{\mathrm{ind.}}{\sim} N(\mu_i, \sigma_y^2)$,  with $\sigma_y^2$ known.

\subsubsection{Full conditional selective inference}

While full CSI enjoys the highest selection quality among the methods considered, this comes at a cost: it can yield very wide (or even infinitely wide) confidence intervals, especially when the signal is weak. This is because full CSI leaves very little information  for inference in the case where the selection event is not ``obvious"
\citep{kivaranovic_length_2021, fithian_optimal_2017, kuchibhotla2022post}.  As the signal strength increases, more information is left over for inference, and the median length of full CSI confidence intervals becomes shorter than that of intervals obtained via sample splitting or data thinning.

\subsubsection{Randomized conditional selective inference}

We end by highlighting the clear advantages of randomized CSI in this setting. 

In randomized CSI, the  information that is unused in Step 1 of Box~\ref{box:tree} is used in Step 2; by contrast, this information is discarded by sample splitting and data thinning.  This has the following implications: (i) For a given level of selection quality, randomized CSI leads to narrower confidence intervals than sample splitting and data thinning. (ii) As the signal strength increases, randomized CSI uses less information for selection (since selection becomes ``obvious"), and therefore retains more for inference; consequently, its intervals narrow. By contrast, sample splitting and data thinning intervals do not adapt to the signal strength.

Furthermore, unlike full CSI, randomized CSI does not suffer from extremely long (much less infinite-length) confidence intervals. This is because randomized CSI places a lower bound on the information available for inference via a user-specified parameter ($\tau^2$ in the proposal of \cite{bakshi2024inference}).

%TC:ignore
\begin{figure}[h]
\includegraphics[width=\textwidth]{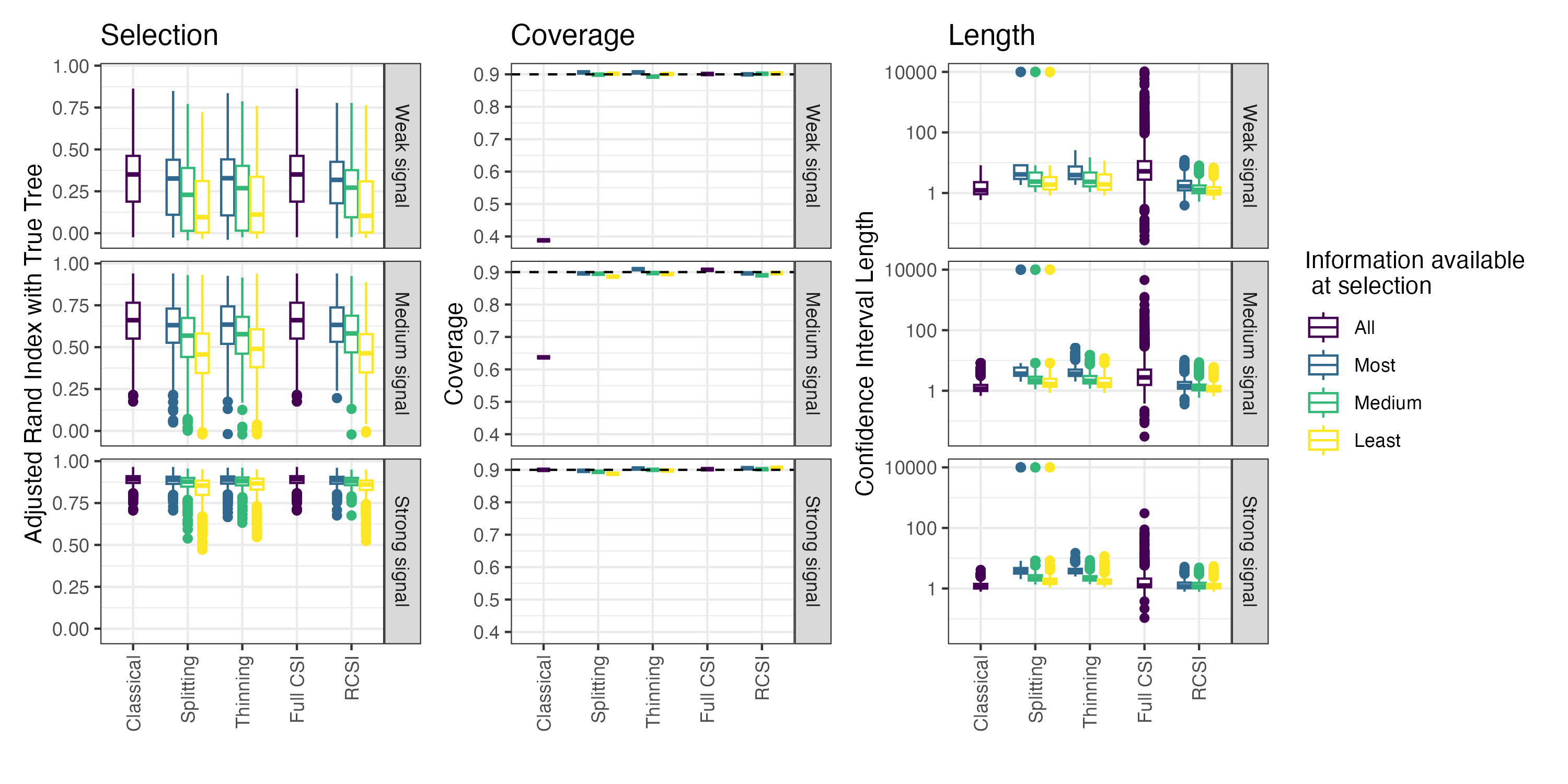}
\vspace{0.1cm}
\caption{The results of the simulation study of Section~\ref{sec_simulation}. For each of 1,000 simulated datasets,  we carry out the methods described in Section~\ref{subsec_compmethods} with $\alpha=0.1$ in Step 2 of Box~\ref{box:tree}, so that the nominal coverage  for all confidence intervals is $90\%$. The left panel shows the selection quality  for each method and each signal strength, as measured using the adjusted Rand index;  values close to $1$ indicate that the regions in the estimated tree nearly match those in the true data-generating model \eqref{eq_truetree}. The center panel shows the coverage of the confidence intervals, and the right panel shows the confidence interval lengths, broken down by method and signal strength. For sample splitting and data thinning, the ``information available at selection" is parameterized by $\epsilon$ (larger values of $\epsilon$ correspond to more information available at selection), whereas for randomized CSI it is parameterized by $\tau^2$ (larger values correspond to less information available at selection). }
\label{fig_simresults}
\end{figure}
%TC:endignore

%% file: sections/sec6_real.tex
\section{Application to Single-Cell RNA Sequencing Data}
\label{sec_realdata}

In this section, we consider Example~\ref{ex_cluster} in the context of single-cell RNA sequencing data. 

\subsection{Problem setup and inferential goals}

Single-cell RNA sequencing is a technology that enables scientists to measure the expression of tens of thousands of genes in tens of thousands of individual cells simultaneously. 
The data can be represented as $X \in \mathbb{Z}_{\geq 0}^{n \times p}$, where $n$ is the number of cells, $p$ is the number of genes, and  $X_{ij}$ is a count corresponding to the number of molecules observed in the $i$th cell that map to the $j$th gene. We assume that 
$
\EE[X_{ij}] = n_i \mu_{ij}$, where $n_i$ represents the sequencing depth of the $i$th cell (usually considered to be a technical artifact),  and $\mu \in \mathbb{R}^{n \times p}$ represents the biological signal of interest \citep{townes2019feature, sarkar2021separating}. 

Our goal is to carry out the procedure in Box~\ref{box:cluster}. 
By contrast to the rest of this paper, our goal is to test a hypothesis, rather than to construct a confidence interval (though of course, there is a direct relationship between the two; see Remark~\ref{remark:hypothesis}).

%we are now considering hypothesis testing instead of confidence intervals. In the hypothesis testing setting, 
 It should be clear to the reader that testing  the hypothesis in Step 2 using the same data used to select the hypothesis in Step 1, without accounting for this double use, will fail to control the selective Type 1 error rate (see Remark~\ref{remark:hypothesis}).
 %. 
%
%the classical approach lets $\Xtrain = \Xtest = X$ and uses a classical test such as a t-test or a Wilcoxon test in Step 2 of Box~\ref{box:cluster} that does not account for the data-driven nature of the parameter. While the invalidity of this classical approach has been 
%pointed out numerous times over the past several years and numerous alternatives have been proposed 
While it is well-understood among statisticians that this approach is problematic \citep{zhang2019valid, song2023clusterde, neufeld2024inference, hivert2024running, courrech2025review},  perhaps surprisingly, it continues to be common practice among biologists: see, for instance, the ``Guided Clustering Tutorial" vignette of the popular software package Seurat \citep{seuratvignette}. 

\begin{mybox}[label=box:cluster]{General procedure for inference after clustering}
%\begin{textbox}[h]\label{box:cluster}
%\section{Box 3: General procedure for inference after clustering}
\begin{enumerate}
       \item Cluster the $n$ cells to obtain two clusters,  $\hat{A}$ and $\hat{B}$, where $\hat{A} \cup \hat{B} =\{1,\ldots,n\}$ and $\hat{A} \cap \hat{B} =\emptyset$. The clusters represent estimated cell types.
\item For  $j=1,\ldots,p$, test the $j$th gene for differential expression: that is, test  $H_{0j} : \bar{\mu}_{\hat{A},j} = \bar{\mu}_{\hat{B},j}$, where $\bar{\mu}_{\hat{A},j} = \frac{1}{|i \in \hat{A}|} \sum_{i \in \hat{A}} \mu_{ij}$ and  $\bar{\mu}_{\hat{B},j}=\frac{1}{|i \in \hat{B}|} \sum_{i \in \hat{B}} \mu_{ij}$.
 \end{enumerate}
%\end{textbox}
\end{mybox}

\begin{remark}
\label{remark:cluster_hard}
Inference after clustering is a particularly challenging statistical problem because sample splitting cannot be applied \citep{gao_selective_2022}. Briefly, suppose we split the rows of $X$ into $\Xtrain$ and $\Xtest$, and cluster the rows of $\Xtrain$. This will not yield cluster assignments for the rows in $\Xtest$.  Furthermore, any attempt to ``transfer" the clusters from $\Xtrain$ to $\Xtest$ to obtain cluster assignments for the cells in $\Xtest$ will necessarily make use of the data in $\Xtest$, thus preventing inference on that data using 
a classical test. This is a subtle point that remains a source of confusion in the scientific literature (see, e.g., \citep{trapnell2024revealing}). 
\end{remark}

\subsection{Dataset and metrics for comparison}

Following \cite{townes2019feature}, we analyze a dataset of peripheral blood mononuclear cells (PBMCs) that was collected by \cite{zheng2017massively}. In addition to conducting single-cell RNA sequencing on the cells, \cite{zheng2017massively} also  sorted the  cells into homogeneous populations using fluorescence activated cell sorting. Thus,  ``ground truth" cell types are known in this dataset, enabling \cite{townes2019feature} to create a ``negative control" and a ``positive control" dataset. Following their lead, we created a negative control dataset containing  $1000$ CD14+ monocytes, and a positive control dataset containing   $1999$ cells identified as either CD14+ monocytes or B cells.

Both the negative and positive control datasets initially contain counts for $p=15,568$ genes, but we reduce the dimension of both datasets to $p=1000$ genes using the deviance feature selection method of \cite{townes2019feature}. While details are given in Appendix~\ref{appendix_real}, we briefly note that we retain the 1000 genes thought to exhibit the most biological variability to construct the positive control dataset, and a set of 1000 genes that are \emph{not} thought to exhibit the most biological variability to construct the negative control dataset.

Since the negative control dataset is composed entirely of CD14+ monocytes and does not include the most variable genes, we expect that the p-values in Step 2 of Box~\ref{box:cluster} will follow a uniform distribution. On the other hand, since  the positive control dataset is composed of two true cell types and includes the most variable genes, we hope to recover the true cell types in Step 1 of Box~\ref{box:cluster}, and to have power to reject null hypotheses in Step 2 of Box~\ref{box:cluster}.
Therefore, in Section~\ref{subsec:realresults}, for each method under consideration, we will report (i) the distribution of Step 2 p-values in the negative control dataset; (ii) the quality of cell type recovery in Step 1 in the positive control dataset; and (iii)  the proportion of rejected null hypotheses in Step 2 in the positive control data set. 

\subsection{Methods for comparison}

We outline our methods for comparison below. We limit ourselves to pre-existing proposals for inference after clustering that achieve conditional guarantees. Recall from Remark~\ref{remark:cluster_hard} that sample splitting, and thus data carving, are not applicable in this setting. We briefly note the existence of several proposals for inference after clustering that do not provide conditional guarantees \citep{kimes2017statistical, song2023clusterde, song2025synthetic}, and are thus omitted.

\subsubsection{Classical method}

To set a baseline for comparison, we begin by describing an incorrect approach, which uses all of the data in both steps of  Box~\ref{box:cluster}, and uses a classical hypothesis test in Step 2. 

While full details are given in Appendix~\ref{appendix_real}, two points are worth noting: (i) In Step 1 of Box~\ref{box:cluster}, we apply k-means clustering with $K=2$ after substantial data preprocessing, which is required to recover the true cell types in the positive control dataset \citep{townes2019feature}. (ii) In Step 2 of Box~\ref{box:cluster}, we test the null hypothesis using standard software for a Poisson GLM with estimates of the sequencing depth as offsets.

\subsubsection{Data thinning}
\label{subsub_thinning}

We consider two versions of data thinning, corresponding to a Poisson assumption and a negative binomial assumption. Each version splits the data $X$ into $\Xtrain$ and $\Xtest$, performs Step 1 of Box~\ref{box:cluster} on $\Xtrain$ in the exact same manner as for the classical method, and conducts Step 2 on $\Xtest$ using standard Poisson or negative binomial GLM software.

\cite{townes2019feature} argue that one can assume $X_{ij} \indsim \text{Poisson}(\lambda_{ij})$. Following \cite{neufeld2024inference}, we apply Poisson 
data thinning to each element of $X$  to obtain $\Xtrain$ and $\Xtest$ with equal allocation of Fisher information. 

Some authors instead argue that scRNA-seq data is overdispersed relative to the Poisson, rendering a negative binomial assumption more appropriate \citep{choudhary2022comparison}.  Following \cite{neufeld2023negative}, we estimate an overdispersion parameter for each gene using the raw data, and apply negative binomial data thinning with these estimated overdispersion parameters and an equal allocation of information between $\Xtrain$ and $\Xtest$. (Negative binomial data thinning with known overdispersion yields $\Xtrain$ and $\Xtest$ that are independent; with an estimated overdispersion, independence is not guaranteed. Accurately estimating the overdispersion can be challenging \citep{hivert2024running}.)

% For both types of thinning, we: (i) pre-process and cluster $\Xtrain$ in the same manner as for the classical method; and (ii) conduct Step 2 using a Poisson or a negative binomial GLM, as in the classical method.  

\subsubsection{Data fission}
\label{subsub_fission}

As mentioned in the previous subsection, negative binomial data thinning yields $\Xtrain$ and $\Xtest$ that are independent only if the 
overdispersion is known. Here, rather than ignore possible dependence between $\Xtrain$ and $\Xtest$ induced by the use of an estimate of overdispersion,  we  instead apply negative binomial data fission to each element of $X$ \citep{leiner2025data}. This does not require knowledge of the overdispersion. 

Interestingly, this operation is equivalent to applying Poisson data thinning to each element $X_{ij}$. 
Thus, we obtain the same $\Xtrain$ and $\Xtest$ as in Poisson data thinning in Section~\ref{subsub_thinning}, and  apply Step 1 of Box~\ref{box:cluster} to $\Xtrain$ as in that subsection. However, 
Step 2 is more challenging due to dependence between $\Xtrain$ and $\Xtest$.  We make use of a conditional likelihood ratio test proposed for this task in \cite{watt2026forthcoming}. 

\subsubsection{Full conditional selective inference}
\label{subsub_fullCSI}

\cite{chen2025testing} propose a full CSI approach for testing for a difference in means of a single feature after clustering. This method applies Step 1 of Box~\ref{box:cluster} to all of the data, and then conducts Step 2 conditional on the clusters selected in Step 1.

However, some challenges arise. The inference in Step 2 requires that the rows of the data follow a multivariate normal distribution with known covariance. Since raw single-cell RNA sequencing data is clearly not multivariate normal, following \cite{chen2025testing} we apply full CSI to a matrix $Z$ where $Z_{ij} = \log\left( \frac{X_{ij}}{\sum_{k=1}^p X_{ik}}+ 1\right)$, and rely on a plug-in estimator of the covariance matrix. The efficient algorithm of \cite{chen2025testing} for computing the p-value in Step 2 requires that the clusters in Step 1 are obtained by applying k-means clustering (or hierarchical clustering) to the full dataset, rather than to a preprocessed version. This means that we cannot use the same preprocessing before k-means clustering as in the other methods; details are in Appendix~\ref{appendix_real}.

\subsection{Results}
\label{subsec:realresults}

\subsubsection{Avoiding false discoveries on the negative control}

The left panel of Figure~\ref{fig:realdata} displays a uniform QQ plot of the 1000 p-values obtained using each method, on the ``negative control" data. As expected, due to the failure to account for selection, the classical method p-values are farthest from uniform. Poisson thinning and negative binomial thinning lead to p-values that are closer to uniform than the classical method, but are still anti-conservative, likely because the distributional assumptions are not met, and so the selection and inference sets are not independent. Negative binomial fission is thus able to improve slightly by accounting for the dependence between the selection and inference sets, but the p-values are still slightly anti-conservative. This may be because the assumptions underlying the method are still not quite met, or because this ``negative control" data is not entirely homogeneous. Despite the fact that the multivariate normality assumption certainly does not hold, the full CSI method seems to lead to p-values that are reasonably close to uniform on this data, mirroring the promising results of \cite{chen2025testing}. However, we note that slight modifications to the procedure for constructing the ``negative control" dataset were quite detrimental to the performance of full CSI.

\subsubsection{Selection quality on the positive control}\label{subsec:ari}

The right panel of Figure~\ref{fig:realdata} reports the adjusted Rand index between the clusters estimated in Step 1 of Box~\ref{box:cluster} and the ground-truth cell types, on the ``positive control" dataset. Since the signal is strong in this data, all methods recover the true cell types with a high degree of accuracy.

\subsubsection{Power to detect differentially-expressed genes on the positive control}

On the ``positive control" dataset, 
the classical method identifies $800$ genes out of $1000$ as highly differentially expressed ($p<0.01$, noting that these p-values do not control the selective Type 1 error and likely overstate the extent of differential expression). Recalling from Remark~\ref{remark:signal} that the effects of selection tend to be minimal when the signal in a dataset is strong, we now ask \emph{how many differentially-expressed genes are identified by each of the other methods with $\alpha=0.01$, and how many of these overlap with the $800$ genes detected by the classical method?} The right panel of Figure~\ref{fig:realdata} shows that all methods considered identify between $644$ and $714$ differentially expressed genes, the majority of which overlap with the classical method.

%TC:ignore
\begin{figure}[h]
\begin{minipage}{0.43\textwidth}
\includegraphics[width=\linewidth]{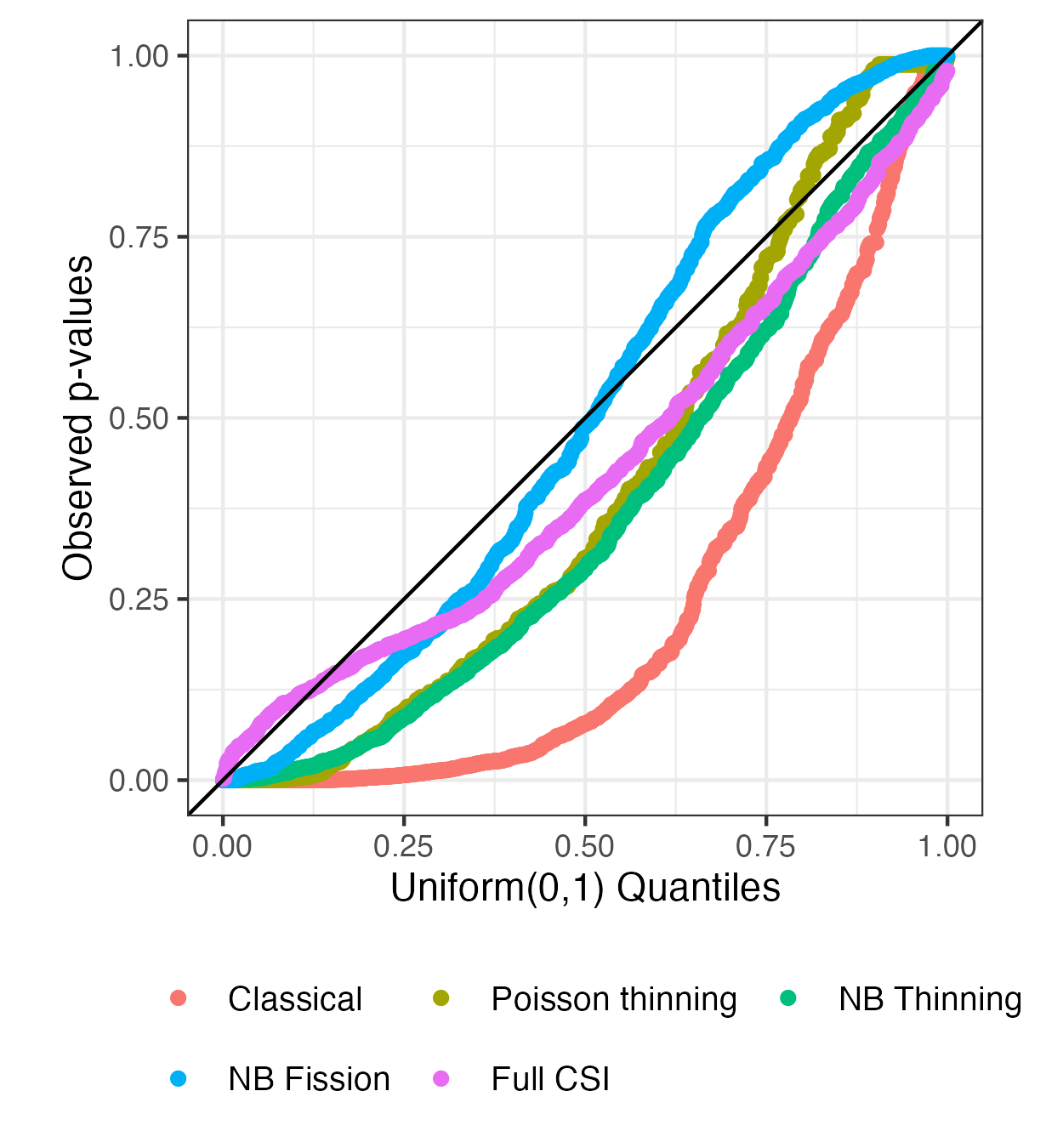}
\end{minipage}
\hfill
\begin{minipage}{0.49\textwidth}
\begin{tabular}{c||c|c|c|}
\hline
\small
Method & Adj.  & $| j : p_j < 0.01|$ & $\big| \{j : p_j < 0.01|\} \cap$  \\
& Rand & & $\{ j : p^{\mathrm{clas.}}_j < 0.01\} \big|$ \\
& Index &  & \\
\hline
\hline
Classical & 0.99 & 800 & 800   \\
\hline
Poisson & 0.99  &  714 & 703 \\
thinning &&& \\ 
\hline
NB  & 0.98 & 693 & 682\\
thinning  &&& \\
\hline
NB Fission  & 0.99 & 644  & 570  \\
\hline
Full CSI &  0.96 & 686 & 633 \\
\hline
\hline
\end{tabular}
\end{minipage}
\caption{The left panel shows uniform QQ plots of the 1000 p-values resulting from each of the five approaches applied to the negative control set. The right panel shows the results of each method applied to the positive control dataset. We report the adjusted Rand index between the true cell types and those estimated by each method. We then report the number of genes identified as highly significant ($p< 0.01$) by each method, as well as the number of these highly significant genes that were also identified by the classical method.
}
\label{fig:realdata}
\end{figure}
%TC:endignore

\subsection{Summary}

We briefly summarize the takeaways of this real data analysis. 
\begin{enumerate}
\item Many approaches to selective inference require distributional assumptions, which cannot be verified on real data. When two selective frameworks rely on \emph{different} distributional assumptions, a direct comparison is difficult. For example, in the right panel of Figure~\ref{fig:realdata}, does Poisson thinning identify more differentially expressed genes than full CSI due to higher power, or because the lack of independence between $\Xtrain$ and $\Xtest$ induced by a violation of the Poisson assumption has resulted in false discoveries? While the left panel of Figure~\ref{fig:realdata} suggests the latter, this cannot be verified. 
\item Data thinning and data fission identify cell types using a \emph{randomized} version of the full dataset $X$. This means that two scientists analyzing the same scRNA-seq dataset with the same pipeline could identify \emph{different} cell types, and thus end up doing inference on \emph{different} parameters. Scientists who are uncomfortable with this additional source of randomness may be tempted to opt for full CSI. 
\item By contrast to data thinning and data fission, an apparent advantage of full CSI is that it does not inject \emph{additional} randomness at the selection stage of an analysis pipeline. However, as the method of \cite{chen2025testing} conditions on the entire solution path of k-means clustering, its performance is surprisingly sensitive to the random seed used to initialize this algorithm. 
\item While \cite{townes2019feature} show that true cell types can be recovered with a high degree of accuracy on this data when k-means is applied after substantial preprocessing, it is more common in the analysis of scRNA-seq data to use graph-based clustering algorithms. For example, the Seurat software package uses the Louvain algorithm \citep{seuratvignette}. In our analysis, we were constrained in our application of full CSI, as we were not able to apply the preprocessing of \cite{townes2019feature} prior to applying k-means. Similarly, we are not aware of tractable implementations of full CSI for clusters selected using a graph-based approach. Thus, a drawback of full CSI is that it places strict limits on the type of clustering that a scientist is able to use.
\item While the goal in creating the ``negative control" dataset 
was to create a dataset with no heterogeneity between the cells, we cannot be certain that this goal was accomplished. That is, it is possible that the ``negative control" dataset still contains biological signal. In fact, slightly different procedures for constructing this ``negative control" dataset led to vastly different QQ plots than the one shown in Figure~\ref{fig:realdata}.
\end{enumerate}

In summary, existing approaches for attaining conditional guarantees rely on strong assumptions and can limit the flexibility of a scientist to analyze their data (e.g. since full CSI for clustering is only tractable for certain clustering algorithms). Despite these drawbacks, the results in Figure~\ref{fig:realdata} are generally promising, with methods tending to yield similar results on both the negative and positive control datasets. 

%% file: sections/sec7_disc.tex
\section{Discussion}
\label{sec_discuss}

Data-driven parameters of interest arise frequently in modern science, leading to a  need for methods that achieve valid inference conditional on selection. 

As we have seen, there are many related approaches for achieving such conditional guarantees, and none of these approaches are uniformly preferable to the others. Any approach involves a tradeoff between the amount of information used for parameter selection and the amount left over for inference. While full CSI, randomized CSI, and data carving are able to make use of \emph{all} of the information in the data, this comes at a cost of increased rigidity and a need for bespoke procedures for inference.
Ultimately, analysts must weigh the competing goals of selection and inference and choose an approach that allocates information appropriately, subject to the constraints of the assumptions that they are willing to make. 

Looking ahead, an important direction is the development of selective inference methods that are practical for scientists to integrate into real data analysis pipelines. To this end, there is a need to develop assumption-lean, flexible selective inference approaches. Some recent proposals avoid strong distributional assumptions in settings with asymptotically Gaussian statistics~\citep{tibshirani_uniform_2018, huang_selective_2024, perry_post-selection_2026}. It is also critically important to improve communication with scientists on (i) the importance of accounting for data-driven selection, (ii) the existence of selective inference approaches, and (iii) how to apply these approaches in the analysis of real data.
For (iii), general-purpose software that can accommodate different analysis pipelines is crucial. In a recent preprint, \cite{miyata2026statistical} aim to provide such software for exactly the application considered in Section~\ref{sec_realdata}; this line of work promises to reduce some of the limitations encountered with full CSI in that section. Similar efforts in different application areas would be a great service to the scientific community. 

Code to reproduce all simulations and real data analyses in this paper is available at \href{https://github.com/anna-neufeld/selective_review}{https://github.com/anna-neufeld/selective\_review}. %github.com/anna-neufeld/selective\_review. 

%% file: sections/zAppendix.tex
\section{Simulation details for Table~\ref{tab:classical-coverages}}
\label{appendix_naive}

In this Appendix, we describe the simulation that leads to the results in Table~\ref{tab:classical-coverages}. We describe the simulation separately for the Examples~\ref{ex_winner}-\ref{ex_cluster}. 

\subsection{Winner's Curse}

We let $Y_i \overset{\mathrm{ind.}}{\sim} \N(\mu_i, \sigma^2)$ for $i = 1,\ldots,n$, where $n=100, \sigma^2=1$, and $\mu_2 = \mu_3 = \ldots = \mu_{100} = 0$. 

As explained in Example~\ref{ex_winner}, for each dataset we (i) select $\hat{k} = \argmax_k Y_k$ and (ii) construct a classical 90\% Z-interval for $\mu_{\hat{k}}$. We generate $2000$ datasets for each of three signal settings:
\begin{itemize}
\item \textbf{None:} $\mu_1=0$, meaning that there is no ``true winner" among the candidates. 
\item \textbf{Medium:} $\mu_1=3.5$, rendering candidate 1 a clear-enough winner that $\hat{k}=1$ with fairly high probability.
\item \textbf{Strong:} $\mu_1=7$, rendering candidate 1 an obvious true winner, such that $\hat{k}=1$ with probability close to one. 
\end{itemize}
For each signal strength, we report the unconditional coverage, i.e. the proportion of the 2,000 classical intervals that contain $\mu_{\hat{k}}$. 

\subsection{Regression Trees}

We let $X \in \mathbb{R}^{n \times p}$ and $Y_i \overset{ind.}{\sim} \N(\mu_i, \sigma^2)$ for $i = 1,\ldots,n$. We let $n=200$, $p=20$, and $\mu_i = 10 + \beta \cdot \bm{1}(x_{i2} > 1)$, and $\sigma^2=1$. For each dataset we fit a one-level CART regression tree \citep{breiman1984classification} to predict $Y_i$ using $X_i$. This regression tree has the effect of dividing our covariate space $\mathbb{R}^p$ into two half-spaces. We let ${\hat{R}}$ denote the half-space with the larger average value of $Y$.  As explained in Example~\ref{ex_tree}, we then construct a classical 90\% Z-interval for $\mu_{\hat{R}}$ based on $\{Y_i :  X_i \in \hat{R}\}$. 

We generate $2000$ datasets for each of three signal strength settings. 
\begin{itemize}
\item \textbf{None:} $\beta=0$, meaning that $\EE[Y]=10$ for all $X$. 
\item \textbf{Medium:} $\beta=1.5$, meaning that there is a true region of covariate space in which $\EE[Y]$ is larger, and that CART will often approximately learn this region. 
\item \textbf{Strong:} $\beta=4$, meaning that there is a true region of covariate space in which $\EE[Y]$ is larger, and that CART will detect this exact region with probability close to $1$.
\end{itemize}
For each signal strength, we report the unconditional coverage, i.e. the proportion of the 2000 confidence intervals that contain $\mu_{\hat{R}}$. 

\subsection{Clustering}

We let $X \in \mathbb{R}^{n \times p}$, with $X_{ij} \overset{\mathrm{ind.}}{\sim} \N(\mu_{ij}, \sigma^2)$, for $i=1,\ldots,200$, $p=1,\ldots,20$, $\sigma^2=1$, and $\mu_{ij} = 10 + \beta \times \bm{1}\left(i \in \{1,\ldots,80\} \cap j \in \{1,\ldots,5\}\right)$.

As in Example~\ref{ex_cluster}, for each dataset we 
(i) apply k-means with $k=2$ to $X$ to partition the observations into two disjoint groups $\hat{A}, \hat{B} \subset \{1,\ldots,n\}$, and then (ii) construct a classical Z-interval for the difference in expression of the first column of $X$ across these two groups;  $\frac{1}{|i \in \hat{A}|} \sum_{i \in \hat{A}} \mu_{i1} - \frac{1}{|i \in \hat{B}|} \sum_{i \in \hat{B}} \mu_{i1}$.

We generate 2,000 datasets for each of three signal settings. 
\begin{itemize}
\item \textbf{None:} $\beta=0$, meaning that the observations are identically distributed and that $\frac{1}{|i \in \hat{A}|} \sum_{i \in \hat{A}} \mu_{i1} = \frac{1}{|i \in \hat{B}|} \sum_{i \in \hat{B}} \mu_{i1}$ for any estimated $\hat{A}$ and $\hat{B}$. 
\item \textbf{Medium:} $\beta=2$, meaning that there are two true clusters and that $\frac{1}{|i \in \hat{A}|} \sum_{i \in \hat{A}} \mu_{i1} - \frac{1}{|i \in \hat{B}|} \sum_{i \in \hat{B}} \mu_{i1} \neq 0$. %for any estimated $\hat{A}$ and $\hat{B}$ that are not orthogonal to the true clusters. 
The signal is strong enough that the estimated clusters tend to approximate the true clusters fairly closely. 
\item \textbf{Strong:} $\beta=4$. The true clusters are well separated, meaning that the estimated clusters are equal to the true clusters with probability close to $1$.
\end{itemize}
In each signal setting, we report the unconditional coverage, i.e. the proportion of the 2000 confidence intervals that cover the true parameter.

\section{Analysis details for Section~\ref{sec_realdata}}
\label{appendix_real}

In this appendix, we provide further details on the data analysis in Section~\ref{sec_realdata}.

\subsection{Initial data preprocessing}
\label{appendix_preproc}

Starting with the raw negative and positive control datasets, we first filter to only include genes that were expressed in at least $20$ cells. We then apply the deviance feature selection method of \cite{townes2019feature} to the full negative and positive control datasets to rank the remaining genes in terms of how far they deviate from a ``null" model that has no biological variability. On the positive control dataset, we retain the 1000 genes that deviate most from this ``null" model; these are the genes that will help us estimate the true cell types and that are likely to be differentially expressed. On the negative control dataset, we retain the genes that are ranked between 1001 and 2000 in terms of their deviation from the null model. While only one cell type is included in the dataset, it is possible that some genes still capture biological variability associated with cell cycle phase or cell subtype; excluding the top 1000 genes helps eliminate this remaining biological variability and thus better ensure that the negative control dataset is a ``null" dataset. (We do not retain the 1000 genes that deviate the \emph{least} from the null model, as these contain mostly zero counts.) 

\subsection{Classical method}

Beginning with the 1000-dimensional negative and positive control datasets obtained after the initial preprocessing of Section~\ref{appendix_preproc}, we carry out the following steps.

We first apply the GLM-PCA method of \cite{townes2019feature} to each dataset to create a 30-dimensional embedding. We then apply k-means with $k=2$ to this 30-dimensional representation of the data to estimate two cell types. Finally, for every gene $X_j$, we fit a Poisson GLM with a log link to predict $X_{ij}$ using an intercept, an indicator for its estimated cell type, and an offset equal to the log of $\hat{n}_i = \sum_{k=1}^p X_{ik})$. We report the standard Wald p-value for the coefficient of the estimated cell type in this model. 

\subsection{Poisson data thinning}

Beginning with the 1000-dimensional negative and positive control datasets obtained after the initial preprocessing of Section~\ref{appendix_preproc}, we first apply Poisson data thinning to every element $X_{ij}$ to obtain $\Xtrain_{ij}$ and $\Xtest_{ij}$. We do this using the R package {countsplit} of \cite{neufeld2024inference}, and set the tuning parameter $\epsilon=0.5$ to allocate equal information to $\Xtrain$ and $\Xtest$. Because the preprocessing in Section~\ref{appendix_preproc} involves all of $X$ (and thus $\Xtest$), it introduces mild double dipping into the data analysis pipeline. However, this is required in order to compare all methods on a common set of $1000$ genes. 

Given $\Xtrain$ and $\Xtest$, we use exactly the same pipeline as the classical method. We apply GLM-PCA to obtain a 30-dimensional embedding of $\Xtrain$, apply k-means to this embedding, and then perform inference on $\Xtest$ using a Poisson GLM. Finally, we report the Wald p-values. 

\subsection{Negative binomial data thinning}

We start by applying negative binomial data thinning to every element $X_{ij}$ of the pre-processed datasets to obtain $\Xtrain_{ij}$ and $\Xtest_{ij}$, once again using the countsplit R package with tuning parameter $\epsilon=0.5$ to allocate equal information to $\Xtrain$ and $\Xtest$. This operation requires an estimate $b_{ij}$ of the overdispersion parameter for each element, under the assumption that $X_{ij} \indsim \mathrm{NB}(n_i \mu_{ij}, b_{ij})$, in the parameterization where $\EE[X_{ij}] = n_i \mu_{ij}$ and $\Var(X_{ij}) = \EE[X_{ij}]+\EE[X_{ij}]^2/b_{ij}$. We estimate overdispersion by applying the  \texttt{vst()} function in the {sctransform} R package \citep{hafemeister2019normalization}, which estimates one overdispersion parameter for each gene (i.e. $\hat{b}_{1j}=\hat{b}_{2j} = \ldots \hat{b}_{nj} = \hat{b}_j$). 

Given $\Xtrain$ and $\Xtest$, we proceed exactly as with Poisson data thinning, though we apply GLM-PCA with a negative binomial assumption and use Wald p-values from a negative binomial GLM rather than a Poisson GLM.

\subsection{Negative binomial data fission}

As mentioned, $\Xtrain$ and $\Xtest$ constructed from the negative binomial data fission procedure of \cite{leiner2025data} coincide with those arising from the Poisson data thinning procedure. We preprocess and cluster $\Xtrain$ in exactly the same way as for Poisson data thinning. 

Given estimated clusters $\hat{A}$ and $\hat{B}$, for each gene $j=1,\ldots,p$, we test the null
$$
H_0: X_{ij} \indsim \mathrm{NB}(n_i \mu_j, b_j) \text{ for } i=1,\ldots,n
$$
against the alternative:
$$
H_A: X_{ij} \indsim 
\begin{cases} 
\mathrm{NB}(n_i \mu_{A_j}, b_j) &\text{ if } i \in \hat{A}, \\
\mathrm{NB}(n_i \mu_{B_j}, b_j) &\text{ if } i \in \hat{B}. \\
\end{cases}
$$
We let $n_i = \sum_{k=1}^p X_{ij}$, but treat these estimated sequencing depths as fixed. We note that this null is slightly different than the one tested with GLMs of the previous methods: we are testing not just for a difference in means across $\hat{A}$ and $\hat{B}$, but rather for a difference in distribution. 

\cite{leiner2025data} provide the form of the conditional distribution of $\Xtest_{ij} \mid \Xtrain_{ij}=\xtrain_{ij}$ after negative binomial fission. Let $\ell(\cdot \mid \xtrain_{ij}, b, \mu, n_i, \epsilon)$ denote the log-likelihood of this conditional distribution. Following \cite{dharamshi2025decomposing}, we carry out a conditional likelihood ratio test. Our likelihood ratio statistic can be written as
$$
-2 \left[ \sup_{\mu, b} \sum_{i=1}^{n} \ell(\Xtest_{ij} \mid \xtrain_{ij}, b, \mu, n_i, \epsilon) - \sup _{\mu, b}   \sum_{i \in A} \ell(\Xtest_{ij}\mid \xtrain_{ij}, b, \mu, n_i, \epsilon) - \sup _{\mu, b} \sum_{i \in B} \ell(\Xtest_{ij} \mid \xtrain_{ij}, b, \mu, n_i, \epsilon) \right].
$$
We approximate this statistic using numerical optimization. Finally, we compare the value of this statistic to the quantiles of a $\chi^2_2$ distribution to obtain a p-value. The theoretical properties of this test, along with the details of the numerical optimization, are explored in forthcoming work \citep{watt2026forthcoming}. 

\subsection{Full conditional selective inference}

We first transform the negative and positive control datasets using standard log-normalization, in which $Z_{ij} = \log\left( \frac{X_{ij}}{\sum_{k=1}^p X_{ik}} + 1 \right)$, in an attempt to make the data approximately follow a normal distribution. We apply k-means with $k=2$ directly to the 1000-dimensional dataset $Z$: this is because the analytical characterization of k-means clustering in \cite{chen2025testing}, required for computational tractability, does not allow for additional preprocessing before clustering. We obtain p-values for each gene after k-means using the R package CADET from \cite{chen2025testing}. This method requires a plug-in estimate of the covariance matrix $\Sigma$ under the assumption that the rows of the matrix $Z$ are independent draws from $\N_p(\mu_i, \Sigma)$. For both the negative and positive control datasets, we assume that $\Sigma$ is a diagonal matrix, as all other methods for comparison assume independence of the genes. For the negative control dataset, we let $\hat{\Sigma}_{jj}$ be the sample variance of the column $Z_j$. For the positive control dataset, we let $\hat{\Sigma}_{jj}$ be the residual variance after accounting for the estimated difference in means across the two estimated clusters.